\documentclass[journal]{IEEEtran}
\usepackage{color}
\usepackage{cite}

\ifCLASSINFOpdf
  \usepackage[pdftex]{graphicx}
\else
  \usepackage[dvips]{graphicx}
\fi
\usepackage[cmex10]{amsmath}
\usepackage{amssymb,amsthm}
\usepackage{algorithm}
\usepackage{algpseudocode}

\usepackage{bm}
\usepackage{mdwmath}
\usepackage{mdwtab}
\usepackage{mathrsfs}   
\usepackage{amsfonts}
\usepackage{booktabs}
\usepackage{multirow}   
\usepackage{enumerate}

\usepackage[tight,footnotesize]{subfigure}
\begin{document}
\title{Principal Component Analysis Based
Broadband Hybrid Precoding for Millimeter-Wave Massive MIMO Systems}
\author{Yiwei Sun,  
        Zhen Gao,~
        Hua Wang,~
        Byonghyo Shim,~\IEEEmembership{Senior Member,~IEEE}, 
        Guan Gui,~
        Guoqiang Mao,~\IEEEmembership{Fellow,~IEEE},
        and Fumiyuki Adachi,~\IEEEmembership{Life Fellow,~IEEE}}

\maketitle

\begin{abstract}
Hybrid analog-digital precoding is challenging
for broadband millimeter-wave (mmWave) massive MIMO systems,
since the analog precoder is frequency-flat but the mmWave channels are frequency-selective.
In this paper, we propose a principal component analysis (PCA)-based broadband hybrid precoder/combiner design, where both the fully-connected array
and partially-connected subarray (including the fixed and adaptive subarrays) are investigated.
Specifically, we first design the hybrid precoder/combiner for fully-connected array and fixed subarray based on PCA,
whereby a low-dimensional frequency-flat precoder/combiner is acquired
based on the optimal high-dimensional frequency-selective precoder/combiner.
Meanwhile, the near-optimality of our proposed PCA approach
is theoretically proven.
Moreover, for the adaptive subarray,
a low-complexity shared agglomerative hierarchical clustering algorithm
is proposed to group the antennas for the further improvement of spectral efficiency (SE) performance.
Besides, we theoretically prove that the proposed antenna grouping algorithm is only determined by
the slow time-varying channel parameters in the large antenna limit.
Simulation results demonstrate the superiority
of the proposed solution over state-of-the-art schemes in SE, energy efficiency (EE), bit-error-rate performance, and the robustness to time-varying channels.
Our work reveals that the EE advantage of adaptive subarray over fully-connected array is obvious for both active and passive antennas,
but the EE advantage of fixed subarray only holds for passive antennas.
\end{abstract}
\renewcommand{\thefootnote}{}
\footnotetext[1]{

Y. Sun, Z. Gao, and H. Wang are with 
Beijing Institute of Technology, Beijing, China
(emails: gaozhen16@bit.edu.cn).
B. Shim is with Institute of New Media and Communications School of Electrical and Computer Engineering, Seoul National University, Seoul, Korea. G. Gui is with College of Telecommunication and Information Engineering, Nanjing University of Posts and Telecommunications, Nanjing, China. G. Mao is with School of Computing and Communication, The University of Technology Sydney, Australia. F. Adachi is with Dept. of Electrical and Communications Engineering Tohoku University, Sendai, Japan. This paper was presented in part at the IEEE GLOBECOM'18 \cite{conf_1,conf_2}.}
\renewcommand{\thefootnote}{\arabic{footnote}}


%
\IEEEpeerreviewmaketitle

\section{Introduction}
Millimeter-wave (mmWave) communication has
been conceived to be a key enabling technology for the next-generation
communications, since it can provide Gbps data rates by leveraging the large transmission bandwidth \cite{R1,R2,R5,ref_plus_7,hTWC,group_1,group_2}. To combat
the severe path loss in mmWave channels, a large number of antennas
are usually employed at both the base stations (BS) and the mobile stations
for beamforming \cite{mao}. However, a large number of antennas could lead to the severe hardware cost and power consumption if each antenna
requires a radio frequency (RF) chain as in conventional fully-digital
MIMO systems. To overcome this problem, hybrid MIMO
has been emerging to trade off hardware
cost with the spectral efficiency (SE) and energy efficiency (EE)~\cite{ref_plus_9,ref_plus_10,R3,R4,R6}.
This tradeoff depends on the specific hybrid MIMO architectures,
which includes the fully-connected array (FCA) and partially-connected subarray (PCS),
and the latter can be further categorized into the fixed subarray (FS) and adaptive subarray (AS) as depicted in Fig.~\ref{FS_AS_fig}~\cite{adp}.
Nevertheless, how to design the hybrid precoding over broadband channels is challenging, as the RF precoding is frequency flat and has the constant-modulus constraint (CMC)~\cite{lim_feedback}.
Therefore, it is of great importance to design an efficient broadband hybrid precoder/combiner for mmWave massive MIMO systems.
\vspace*{-4mm}
\subsection{Related Work}
Narrowband hybrid precoding has been investigated in \cite{OMP,CS_lowcomplx,par_13,ref_plus_9,ref_plus_10,hy-bd,mao7}. Specifically, a compressive sensing (CS)-based hybrid precoding was proposed in \cite{OMP}, where the channel sparsity was exploited
with the aid of orthogonal matching pursuit (OMP) algorithm.
To reduce the computational complexity, a low-complexity CS-based beamspace hybrid precoding was developed in \cite{CS_lowcomplx}.
To improve the EE, an iterative analog precoding was designed for FS \cite{par_13}, but it failed to consider the digital precoding.
Moreover, a constant envelope hybrid precoding scheme was proposed, where the hybrid precoding was designed under the per-antenna constant envelope constraints \cite{ref_plus_9}. Additionally, the codebook-based scheme, the hybrid block diagonal scheme, and the heuristic scheme were respectively proposed in \cite{ref_plus_10}, \cite{hy-bd}, and \cite{mao7} for multi-user MIMO.
In contrast to the described prior art,
hybrid solutions with significantly different
analog architectures have also been explored
in \cite{major1,major2}, which provide benefits of reduced
hardware cost and channel estimation overhead.
However, \cite{OMP,CS_lowcomplx,par_13,ref_plus_9,ref_plus_10,hy-bd,mao7} only assumed the flat fading channels. 

To effectively combat the time dispersive channels, several elegant broadband hybrid precoding solutions
have been proposed in \cite{Kim_2,lim_feedback,170825,fre_flat}. Most of them adopted OFDM so that
the broadband frequency-selective fading channels
were converted into multiple parallel narrowband frequency-flat fading channels.
To be specific, a hybrid precoding scheme has been proposed in \cite{Kim_2} to support
single stream transmission in MIMO-OFDM system, where the optimal beam pair was exhaustively searched from a codebook predefined for FS.
To reduce the computational complexity, the limited-feedback codebook based broadband hybrid precoder has been proposed for fully-connected array (FCA) \cite{lim_feedback}.
Moreover, by exploiting the channel correlation information among different subcarriers, a broadband hybrid precoding was proposed, where both FCA and PCS were investigated \cite{170825}.
However, \cite{lim_feedback} did not specify the combiner design at the receiver, and \cite{170825} assumed the fully-digital array at the receiver. Note that in \cite{170825}, although a greedy algorithm was proposed to group the antennas for AS,
this method could suffer from the poor performance due to the unbalanced antenna grouping.
Besides, by proving the dominant subspaces of frequency domain channel matrices at different subcarriers are equivalent, \cite{fre_flat} has theoretically revealed the optimality of the frequency-flat precoding. However, this conclusion was based on the ideal sparse channels by assuming the discrete angles of arrival/departure (AoA/AoD), and the specific precoder/combiner solution was not explicitly provided.

\vspace*{-4mm}
\subsection{Our Contributions}

\begin{figure}[tb]
    \centering
	\includegraphics[width=\columnwidth, keepaspectratio]{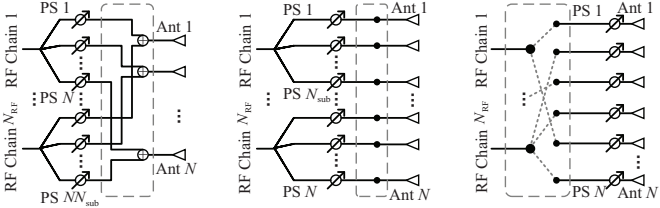}\vspace*{-6mm}
	\caption{Different hybrid MIMO architectures: (a) FCA; (b) FS; (c) AS, where the
		connection between phase shifters and antennas is adaptive, and only one possible connection is shown
		in this figure. ``PS" denotes phase shifter and ``Ant" denotes antenna.}\label{FS_AS_fig}\vspace*{-6mm}
\end{figure}

In this paper, we propose a broadband hybrid precoding for mmWave massive MIMO systems,
where both FCA and PCS are investigated. 
First, the hybrid
precoder/combiner for FCA is designed based on the principal component analysis (PCA),
whereby a near-optimal low-dimensional analog precoder/combiner is acquired from
the optimal high-dimensional fully-digitally precoder/combiner.
Furthermore, this PCA-based approach is generalized to PCS. 
Besides, for the AS, a low-complexity shared agglomerative hierarchical clustering (shared-AHC) algorithm is proposed to group the antennas
adapted to the channels for the further improved SE performance.
The contributions of this paper are summarized as follows:

\begin{itemize}
    \item \textbf{Near-optimal PCA-based hybrid precoder/combiner design}.
    Based on the framework of PCA (weighted PCA), we design the analog precoder (combiner) at the transmitter (receiver)
    according to the optimal fully-digital precoder (combiner).
    By contrast, state-of-the-art solutions usually only design the hybrid precoder,
    nevertheless, the hybrid combiner was not specified and the BER performance was not evaluated \cite{Kim_2,lim_feedback,170825,fre_flat}.
    We theoretically prove the near-optimality of our proposed PCA-based solution, whose SE and BER advantages over state-of-the-art solutions are also verified by simulations.
    %
	\item \textbf{Low-complexity shared-AHC algorithm to group antennas for AS}. The optimal antenna grouping for AS requires the exhaustive search, which
 suffers from the prohibitively high computational complexity. To solve this problem,
 we formulate the antenna grouping problem as the clustering problem in machine learning, and further
propose a low-complexity shared-AHC algorithm.
Meanwhile, we prove that the antenna grouping strategy based on the proposed shared-AHC algorithm
is only determined by the slow time-varying channel parameters in the
large antenna limits.
By comparison, the existing solution \cite{170825}
only considered the antenna grouping at the transmitter (TX), and it could lead to the extremely unbalanced antenna grouping case that no antenna was assigned to one RF chain.

	\item \textbf{EE performance evaluation in practical passive/active antennas}. Passive and active antennas have the different array architectures, which result in the different numbers of power-consuming electronic elements (e.g., power amplifiers).
EE analysis in prior work \cite{EE1,EE2,EE3,niu45} did not distinguish the passive and active antennas. By contrast, we consider the practical passive/active antennas for EE performance analysis.
%
%
Our work demonstrates that the EE advantage
of AS over FCA is overwhelming for both active and passive antennas,
while the EE advantage of FS over FCA can only be observed for passive antennas.
\end{itemize}

The rest of this paper is organized as follows. The system model is introduced in Section II. The proposed PCA-based hybrid precoder/combiner for FCA is presented in Section III.
The proposed PCA-based hybrid precoder/combiner for PCS and the proposed shared-AHC-based antenna grouping for AS are presented in Section IV.
In Section V, we evaluate the system performance. 
Finally, we conclude this paper in Section VII. This paper was presented in part at the IEEE GLOBECOM'18 \cite{conf_1,conf_2}. Except for the work presented in \cite{conf_1,conf_2}, the unique contribution of this paper is the expansion of the PCS structure on hybrid combiner and the evaluation of the performance, including the EE performance of the system, the computational
complexity of the antenna grouping algorithm, the robustness of antenna grouping algorithm to
time-varying channel, and the robustness to channel perturbation.

\textsl{Notations}: Following notations are used throughout this paper. $\mathbf{A}$ is a matrix, $\mathbf{a}$ is a vector, $a$ is a scalar, and $\mathcal{A}$ is a set.
Conjugate transpose and transpose of $\mathbf{A}$ are $\mathbf{A}^H$ and $\mathbf{A}^T$, respectively.
The $(i,j)$th entry of $\mathbf{A}$ is $[\mathbf{A}]_{i,j}$, $[\mathbf{A}]_{i,:}$ ($[\mathbf{A}]_{:,j}$) denotes
the $i$th row ($j$th column) of $\mathbf{A}$, sub-matrix $[\mathbf{A}]_{i_1:i_2,:}$ ($[\mathbf{A}]_{:,j_1:j_2}$) consists of
the $i_1$th to $i_2$th rows ($j_1$th to $j_2$th columns) of $\mathbf{A}$, and sub-matrix $[\mathbf{A}]_{i_1:i_2,j_1:j_2}$ consists of
the $i_1$th to $i_2$th rows and $j_1$th to $j_2$th columns of $\mathbf{A}$. Frobenius norm, $\ell_2$-norm, and determinant are denoted by $||\cdot||_F$,
$||\cdot||_2$, and $\det(\cdot)$, respectively. 
$\text{card}(\mathcal{A})$ is the
cardinality of a set $\mathcal{A}$. $|\bf{A}|$, $\angle(\mathbf{A})$, and  $\mathfrak{R}\{\mathbf{A}\}$ are
matrices whose elements are the modulus values, phase
values, and real parts of the corresponding elements in $\mathbf{A}$, respectively. $\text{round}(\mathbf{A})$ is a matrix by replacing every element in $\mathbf{A}$ with
its closest integer. 
$\otimes$ represents the Kronecker product. Subtraction between sets $\mathcal{A}$ and $\mathcal{B}$ is $\mathcal{A}\setminus\mathcal{B}=\{x|x\in\mathcal{A}~\&\notin\mathcal{B}\}$. $\mathbf{I}_N$ denotes an
identity matrix with size $N\times N$. The $i$th largest singular value of a matrix $\mathbf{A}$ is defined as $\lambda_i(\mathbf{A})$.
$(\cdot)^+$ denotes $(a)^+=a$ if $a>0$, otherwise $(a)^+=0$.
$\text{blkdiag}(\mathbf{a}_1,\cdots,\mathbf{a}_K)$ is a block diagonal matrix with $\mathbf{a}_i$ ($1\leq i\leq K$) on its
diagonal blocks.
Finally, $Col\mathbf{A}$ is the column space of the matrix $\mathbf{A}$, and
$(Col\mathbf{A})^\bot$ is the orthogonal complement space of $Col\mathbf{A}$.
\section{System Model}
\begin{figure}[tb]
  \centering
  \includegraphics[width=\columnwidth, keepaspectratio]{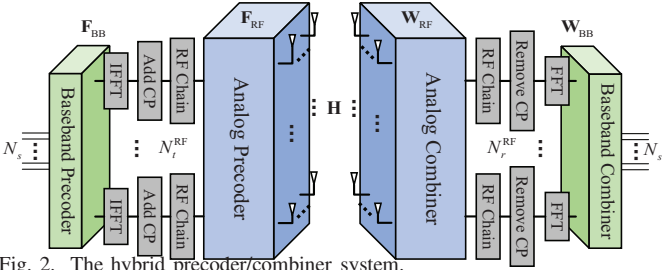}\vspace*{-6mm}
  \caption{The hybrid precoder/combiner system.}\label{fig_sys}\vspace*{-6mm}
\end{figure}
We consider a downlink mmWave massive MIMO system with perfect synchronization
(this can be achieved by efficient algorithms like \cite{per_syn,xiao46}) as shown in Fig. \ref{fig_sys}, where both the BS and the user employ the uniform planar array (UPA), and OFDM is adopted to combat the frequency-selective fading channels. The BS is equipped with $N_t=N_t^v\times N_t^h$ antennas and $N_t^{\rm RF}\ll N_t$ RF chains, where $N_t^v$ and $N_t^h$ are the numbers of vertical and horizontal transmit antennas, respectively. The user is equipped with $N_r=N_r^v\times N_r^h$ antennas and $N_r^{\rm RF}\ll N_r$ RF chains, where $N_r^v$ and $N_r^h$ are the numbers of vertical and horizontal receive antennas, respectively. In the downlink, the received signals of the $k$th subcarrier at the user are \cite{OMP}
\begin{equation}\label{rcv_sig}
\begin{aligned}
\mathbf{r}[k]=&(\mathbf{W}_{\rm RF}\mathbf{W}_{\rm BB}[k])^H(\mathbf{H}[k]\mathbf{F}_{\rm RF}\mathbf{F}_{\rm BB}[k]\mathbf{x}[k]+\mathbf{n}[k]),
\\&1\leq k\leq K,
\end{aligned}
\end{equation}
where $K$ is the number of subcarriers, $\mathbf{F}_{\rm BB}[k]\in\mathbb{C}^{N_t^{\rm RF}\times N_s}$, $\mathbf{F}_{\rm RF}\in\mathbb{C}^{N_t\times N_t^{\rm RF}}$, $\mathbf{W}_{\rm BB}[k]\in\mathbb{C}^{N_r^{\rm RF}\times N_s}$, $\mathbf{W}_{\rm RF}\in\mathbb{C}^{N_r\times N_r^{\rm RF}}$, $\mathbf{H}[k]\in\mathbb{C}^{N_r\times N_t}$, $\mathbf{x}[k]\in\mathbb{C}^{N_s\times 1}$, and $\mathbf{n}[k]\in\mathbb{C}^{N_r\times 1}$ are the digital precoder, analog precoder, digital combiner, analog combiner, frequency-domain channel matrix, transmitted signal, and additive white Gaussian noise (AWGN) associated with the $k$th subcarrier, respectively, and $N_s$ is the number of data streams. $\mathbf{n}[k]\sim \mathcal{CN}(0,\sigma_n^2\mathbf{I}_{N_r})$ and $\mathbf{x}[k]$ satisfies $\mathbb{E}[\mathbf{x}[k]\mathbf{x}^H[k]]=\mathbf{I}_{N_s}$.
Here $\mathbf{H}[k]=\sum_{d=0}^{D-1}\mathbf{\tilde H}[d]e^{-j\frac{2\pi k}{K}d}$, where $D$ is the maximum delay spread of the discretized channels, and $\mathbf{\tilde H}[d]\in\mathbb{C}^{N_r\times N_t}$ is the delay-domain channel matrix of the $d$th delay tap. We consider the clustered channel model \cite{OMP}, where the channel comprises $N_{\rm cl}$ clusters of multipaths with $N_{\rm ray}$ rays in each cluster. Thus, the delay-domain channel matrix is
\begin{equation}\label{H_d}
  \mathbf{\tilde H}[d]=\sum\nolimits_{i=1}^{N_{\rm cl}}\sum\nolimits_{l=1}^{N_{\rm ray}}{\tilde p}_{i,l}[d]\mathbf{a}_r(\phi^r_{i,l},\theta^r_{i,l})\mathbf{a}_t^H(\phi^t_{i,l},\theta^t_{i,l}),
\end{equation}
where ${\tilde p}_{i,l}[d]=\sqrt{N_tN_r/(N_{\rm cl}N_{\rm ray})}\alpha_{i,l}p(dT_s-\tau_{i,l})$ is the delay-domain channel coefficient,
$\tau _{i,l}$,  $\alpha _{i,l}$, and $p(\tau)$ are the delay, the complex path gain, and the pulse shaping filter for $T_s$-spaced signaling,
respectively. Thus the frequency-domain channel coefficient is $p_{i,l}[k]\!\!=\!\!\sum_{d=0}^{D-1}{\tilde p}_{i,l}[d]e^{-j\frac{2\pi k}{K}d}$.
In (\ref{H_d}), $\mathbf{a}_t(\phi^t_{i,l},\theta^t_{i,l})$ and $\mathbf{a}_r(\phi^r_{i,l},\theta^r_{i,l})$ are the steering vectors of the $l$th
path in the $i$th cluster at the TX and receiver (RX), respectively, where $\phi^t_{i,l}$ ($\phi^r_{i,l}$) and $\theta^t_{i,l}$ ($\theta^r_{i,l}$) are the azimuth and elevation
angles of the $l$th ray in the $i$th cluster for AoDs (AoAs).
The steering vector of the $l$th ray in the $i$th cluster for AoD is
$\mathbf{a}_t(\phi ^t_{i,l},\theta^t_{i,l})\!=\!\mathbf{e}_t^v(\Omega^v_{i,l})\otimes \mathbf{e}_t^h(\Omega^h_{i,l})$,
where $\mathbf{e}_t^v(\Omega^v_{i,l})\!=\!\frac{1}{\sqrt{N_t^v}}
\begin{bmatrix}1\ e^{-j2\pi\Omega^v_{i,l}}\
\cdots\ e^{-j2\pi(N_t^v-1)\Omega^v_{i,l}}\end{bmatrix}\!^T$, $\mathbf{e}_t^h(\Omega^h_{i,l})\!=\!
\frac{1}{\sqrt{N_t^h}}\begin{bmatrix}1\ e^{-j2\pi\Omega^h_{i,l}}\
\cdots\ e^{-j2\pi(N_t^h-1)\Omega^h_{i,l}}\end{bmatrix}\!^T$, $\Omega^h_{i,l}=\frac{\sin(\theta ^t_{i,l})\sin(\phi ^t_{i,l})d_h}{\lambda_c}$, $\Omega^v_{i,l}=
\frac{\cos(\theta ^t_{i,l})d_v}{\lambda_c}$, ${\lambda_c}$ is the carrier wavelength, and $d_h=\frac{\lambda_c}{2}$, $d_v=\frac{\lambda_c}{2}$ denote
the horizontally and vertically antenna spacing, respectively \cite{mao}.
Similarly, the receive steering vectors are $\mathbf{a}_r(\phi^r_{i,l},\theta^r_{i,l})\!=\!
[1 \cdots  e^{-j2\pi(m\Psi^h_{i,l}\!+\!n\Psi^v_{i,l})}  \\ \cdots e^{-j2\pi((N_r^h-1)\Psi^h_{i,l}\!+\!(N_r^v-1)\Psi^v_{i,l})}]\!^T/\sqrt{N_r}$, for $1\leq i\leq N_{\rm cl}$, $1\leq l\leq N_{\rm ray}$, where $\Psi^h_{i,l}=\frac{\sin(\theta ^r_{i,l})\sin(\phi ^r_{i,l})d_h}{\lambda_c}$ and $\Psi^v_{i,l}=\frac{\cos(\theta ^r_{i,l})d_v}{\lambda_c}$.

The achievable SE for the mmWave MIMO-OFDM system can be expressed as \cite{lim_feedback}
\begin{equation}\label{SE}
\begin{aligned}
R=\!&\frac{1}{K}\!\!\sum_{k=1}^{K}\!\log_2(\det(\mathbf{I}\!+\!
\mathbf{R}_n^{-1}\![k]\mathbf{W}_{\rm BB}^H\![k]\mathbf{W}_{\rm RF}^H
\mathbf{H}[k]\mathbf{F}_{\rm RF}\mathbf{F}_{\rm BB}\![k]
\\\times&\mathbf{F}_{\rm BB}^H\![k]\mathbf{F}_{\rm RF}^H\mathbf{H}^H\![k]
\mathbf{W}_{\rm RF}\!\mathbf{W}_{\rm BB}[k])),\!\!\!
\end{aligned}
\end{equation}
where $\mathbf{R}_n[k]=\sigma_n^2\mathbf{W}_{\rm BB}^H[k]\mathbf{W}_{\rm RF}^H
\mathbf{W}_{\rm RF}\mathbf{W}_{\rm BB}[k]$, $\sum_{k=1}^K||\mathbf{F}_{\rm RF}\mathbf{F}_{\rm BB}[k]||_F^2=KN_s$, $[\mathbf{F}_{\rm RF}]_{:,i}\in\mathcal{F}_{\rm RF}$ for $1\leq i\leq N_t^{\rm RF}$, $[\mathbf{W}_{\rm RF}]_{:,j}\in\mathcal{W}_{\rm RF}$ for $1\leq j\leq N_r^{\rm RF}$, $\mathcal{F}_{\rm RF}\subseteq\mathbb{C}^{N_t\times1}$ and $\mathcal{W}_{\rm RF}\subseteq\mathbb{C}^{N_r\times1}$ are respectively the sets of feasible RF precoder and combiner satisfying the CMC for each entry. Note that our work is distinctly different from the previous work \cite{170825}, which considered the hybrid precoder  but the fully-digital combiner. In this paper, we consider the hybrid MIMO architecture at both the TX and RX. Given full channel state information (CSI) at both TX and RX by using state-of-art efficient channel estimation solutions \cite{R1,ref_plus_7,hTWC,group_2,chanel1}, our goal is to design the hybrid precoder and combiner that maximize the SE. Since the sum rate $R$ is a function of variables ($\mathbf{F}_{\rm RF}$,$\{\mathbf{F}_{\rm BB}[k]\}_{k=1}^K$,$\mathbf{W}_{\rm RF}$,$\{\mathbf{W}_{\rm BB}[k]\}_{k=1}^K$), it is computationally inefficient to jointly optimize the sum rate. How to solve this intractable problem will be detailed as follows.
\section{PCA-Based Hybrid Precoder/Combiner for FCA}
In this section, we will propose the hybrid precoder/combiner design for mmWave massive MIMO systems based on FCA (weighted FCA), whereby the frequency-flat RF precoder (combiner) can be acquired from the optimal fully-digital frequency-selective precoder (combiner).
\subsection{PCA-Based Hybrid Precoder Design at TX}
We first design the hybrid precoder by solving the following optimization problem
\begin{equation}
\begin{aligned}\label{opt}
(\mathbf{F}_{\rm RF}^{\rm opt},\{\mathbf{F}_{\rm BB}^{\rm opt}[k]\}_{k=1}^K)\!=&\!\!\!\!\!\!\max\limits_{\mathbf{F}_{\rm RF},\{\mathbf{F}_{\rm BB}[k]\}_{k=1}^K}\sum\nolimits_{k=1}^{K}\log_2(\det(\mathbf{I}_{N_r}\!\!+\!\!\frac{1}{\sigma_n^2}\\
\times&\mathbf{H}[k]\mathbf{F}_{\rm RF}\mathbf{F}_{\rm BB}[k]\mathbf{F}_{\rm BB}^H[k]\mathbf{F}_{\rm RF}^H\mathbf{H}^H[k]))
\\\text{s.t. }&[\mathbf{F}_{\rm RF}]_{:,i}\in\mathcal{F}_{\rm RF}\text{ for }1\leq i\leq N_t^{\rm RF}, \\&\sum\nolimits_{k=1}^K||\mathbf{F}_{\rm RF}\mathbf{F}_{\rm BB}[k]||_F^2=KN_s.
\end{aligned}
\end{equation}
The joint optimization of $\mathbf{F}_{\rm RF}$ and $\{\mathbf{F}_{\rm BB}[k]\}_{k=1}^K$ in (\ref{opt}) is
difficult due to the coupling between $\mathbf{F}_{\rm RF}$ and $\{\mathbf{F}_{\rm BB}[k]\}_{k=1}^K$. This motivates us to design $\mathbf{F}_{\rm RF}$ and $\{\mathbf{F}_{\rm BB}[k]\}_{k=1}^K$, separately. We consider $\mathbf{\widetilde{F}}_{\rm BB}[k]=(\mathbf{F}_{\rm RF}^H\mathbf{F}_{\rm RF})^{\frac{1}{2}}\mathbf{F}_{\rm BB}[k]$ to be the equivalent baseband precoder, so that
(\ref{opt}) can be rewritten as
\begin{equation}
\begin{aligned}\label{opt_RF}
&\max\limits_{\mathbf{F}_{\rm RF},\{\mathbf{\widetilde{F}}_{\rm BB}[k]\}_{k=1}^K}\!\sum_{k=1}^{K}\log_2(\det(\mathbf{I}_{N_r}\!\!+\!\frac{1}{\sigma_{ n}^2}
\mathbf{H}[k]\mathbf{F}_{\rm RF}(\mathbf{F}_{\rm RF}^H\mathbf{F}_{\rm RF}\!)^{-\frac{1}{2}}
\\&\times\mathbf{\widetilde{F}}_{\rm BB}[k]\mathbf{\widetilde{F}}_{\rm BB}^H[k](\mathbf{F}_{\rm RF}^H\mathbf{F}_{\rm RF}\!)^{-\frac{1}{2}}\mathbf{F}_{\rm RF}^H\mathbf{H}^H\![k]))
\\&\text{s.t. }[\mathbf{F}_{\rm RF}]_{:,i}\in\mathcal{F}_{\rm RF}\text{ for }1\leq i\leq N_t^{\rm RF},\\&\sum\nolimits_{k=1}^K||\mathbf{\widetilde{F}}_{\rm BB}[k]||_F^2=KN_s.
\end{aligned}
\end{equation}
To solve the optimization problem (\ref{opt_RF}), we first investigate the optimal solution of $\{\mathbf{\widetilde{F}}_{\rm BB}[k]\}_{k=1}^K$. Specifically, we consider the singular value decomposition (SVD) of $\mathbf{H}[k]$ at the $k$th subcarrier
\begin{equation}\label{SVD}
\setlength{\abovedisplayskip}{5pt}
\setlength{\belowdisplayskip}{5pt}
\mathbf{H}[k]=\mathbf{U}[k]\mathbf{\Sigma}[k]\mathbf{V}^H[k],
\end{equation}
and the SVD of the matrix $\mathbf{\Sigma}[k]\mathbf{V}^H[k]\mathbf{F}_{\rm RF}(\mathbf{F}_{\rm RF}^H\mathbf{F}_{\rm RF})^{-1/2}=\mathbf{\widetilde{U}}[k]\mathbf{\widetilde{\Sigma}}[k]
\mathbf{\widetilde{V}}^H[k]$. Therefore, the optimal equivalent baseband precoder is $\mathbf{\widetilde{F}}_{\rm BB}[k]=[\mathbf{\widetilde{V}}[k]]_{:,1:N_s}\mathbf{\Lambda}[k]\in\mathbb{C}^{N_t^{\rm RF}\times N_s}$, and thus the optimal baseband precoder $\mathbf{F}_{\rm BB}[k]$ can be expressed as
\begin{equation}\label{F_BB}
\mathbf{F}_{\rm BB}[k]\!\!=\!\!(\mathbf{F}_{\rm RF}^H\mathbf{F}_{\rm RF})^{-\frac{1}{2}}\mathbf{\widetilde{F}}_{\rm BB}[k]
\!\!=\!\!(\mathbf{F}_{\rm RF}^H\mathbf{F}_{\rm RF})^{-\frac{1}{2}}[\mathbf{\widetilde{V}}[k]]_{:,1:N_s}\!\mathbf{\Lambda}[k],
\end{equation}
where $\mathbf{\Lambda}[k]\in\mathbb{C}^{N_s\times N_s}$ is the water-filling solution matrix, i.e.,
\begin{equation}\label{wf_1}
[\mathbf{\Lambda}[k]]_{i,i}^2=(\mu -\sigma_{\rm n}^2/[\mathbf{\widetilde{\Sigma}}[k]]_{i,i}^2)^+,1\leq i\leq N_s,1\leq k\leq K,
\end{equation}
$\mu$ meets $\sum_{k=1}^K\sum_{i=1}^{N_s}(\mu-\sigma_{\rm n}^2/
[\mathbf{\widetilde{\Sigma}}[k]]_{i,i}^2)^+=KN_s$.
Then the joint optimization of $\mathbf{F}_{\rm RF}$ and $\{\mathbf{F}_{\rm BB}[k]\}_{k=1}^K$ is simplified as the optimization of $\mathbf{F}_{\rm RF}$ in (\ref{opt_RF}).
Considering a conventional frequency-selective precoder as the optimal fully-digital precoder\footnote{Water filling is not considered for the ease of analysis, and it can be adopted to maximize the sum rate according to \cite{fre_flat}.} ${\color{red}\mathbf{F}_{\rm FD}^{\rm opt}}[k]=[\mathbf{V}[k]]_{:,1:N_s}\in\mathbb{C}^{N_t\times N_s}$, we further have the approximation in Lemma 1.
\newtheorem{lem}{\textbf{Lemma}}
\begin{lem}
    Optimization problem in (\ref{opt}) can be approximately written as
    \begin{equation}\label{opt_final}
        \begin{aligned}
            &\max_{\mathbf{F}_{\rm RF}}\sum_{k=1}^{K}
            ||{\mathbf{F}_{\rm FD}^{\rm opt}}^H[k]\mathbf{F}_{\rm RF}\mathbf{F}_{\rm BB}[k]||_F^2
            \\&\text{\rm s.t. }[\mathbf{F}_{\rm RF}]_{:,i}\in\mathcal{F}_{\rm RF}\text{ for }1\leq i\leq N_t^{\rm RF}.
        \end{aligned}
    \end{equation}
    The optimization problem (\ref{opt}) is equivalent to (\ref{opt_final}) when the following requirements are reached:
    \begin{enumerate}
        \item Hybrid precoder $\mathbf{F}_{\rm RF}\mathbf{F}_{\rm BB}[k]$ can be sufficiently ``close" to the optimal fully-digital precoder. 
        \item The signal-to-noise ratio (SNR) is sufficiently high.
    \end{enumerate}
\end{lem}
\begin{proof}
	See Appendix A.
\end{proof}
 Note that the frequency domain MIMO channel matrices $\{\mathbf{H}[k]\}_{k=1}^K$ have the
 same column/row space \cite{fre_flat}, and $\mathbf{F}_{\rm RF}$ is identical for all subcarriers.
 So $\mathbf{F}_{\rm RF}$ can be regarded as a representation of such a column space. This observation motivates us to
 design the RF precoder $\mathbf{F}_{\rm RF}$ under the framework of PCA \cite{mach_lern}, whereby the principal components constituting $\mathbf{F}_{\rm RF}$ can be acquired from the data set matrix ${\widetilde{\mathbf{F}}_{\rm FD}^{\rm opt}}=\begin{bmatrix}{\mathbf{F}_{\rm FD}^{\rm opt}}[1]\ {\mathbf{F}_{\rm FD}^{\rm opt}}[2]\ \cdots\ {\mathbf{F}_{\rm FD}^{\rm opt}}[K]\end{bmatrix}$. To achieve the stable solution with low complexity for PCA problem,
we consider the SVD approach to process ${\widetilde{\mathbf{F}}_{\rm FD}^{\rm opt}}$ \cite{mach_lern} as follows.
\newtheorem{prop}{\textbf{Proposition}}
\begin{prop}
    Given ${\widetilde{\mathbf{F}}_{\rm FD}^{\rm opt}}$ and its SVD ${\widetilde{\mathbf{F}}_{\rm FD}^{\rm opt}}=\mathbf{U}_{\widetilde{\mathbf{F}}_{\rm FD}^{\rm opt}}\mathbf{\Sigma}_{\widetilde{\mathbf{F}}_{\rm FD}^{\rm opt}}\mathbf{V}_{\widetilde{\mathbf{F}}_{\rm FD}^{\rm opt}}^H$, the sub-optimal solution to (\ref{opt_final}) can be expressed as $\mathbf{F}_{\rm RF}=\frac{1}{{\sqrt {{N_t}} }}{\rm exp}(j\angle([\mathbf{U}_{\widetilde{\mathbf{F}}_{\rm FD}^{\rm opt}}]_{:,1:N_t^{\rm RF}}))$.
\end{prop}
\begin{proof}
	See Appendix B.
\end{proof}

\begin{algorithm}[t]
    \caption{ PCA-based RF Precoder Design.}
	\label{alg:PCA_AS}
	\begin{algorithmic}[1]
		\renewcommand{\algorithmicrequire}{\textbf{Input:}}
		\renewcommand\algorithmicensure {\textbf{Output:}}
		\Require
		$\{{\mathbf{F}_{\rm FD}^{\rm opt}}[k]\}_{k=1}^K$,
		$N_t^{\rm RF}$, $N_t$, and $Q$.
		\Ensure
		$\mathbf{F}_{\rm RF}$.
		\State ${\widetilde{\mathbf{F}}_{\rm FD}^{\rm opt}}\!=\!\begin{bmatrix}{\mathbf{F}_{\rm FD}^{\rm opt}}[1]\ {\mathbf{F}_{\rm FD}^{\rm opt}}[2]\ \!\cdots\ \! {\mathbf{F}_{\rm FD}^{\rm opt}}[K]\end{bmatrix}$
(Here ${\mathbf{F}_{\rm FD}^{\rm opt}}[k]\!=\![\mathbf{U}[k]]_{:,1:N_s}$, $\mathbf{H}[k]\!=\!\mathbf{U}[k]\mathbf{\Sigma}[k]\mathbf{V}^H[k]$, $\forall k$)
		\State Apply SVD to ${\widetilde{\mathbf{F}}_{\rm FD}^{\rm opt}}$, i.e., ${\widetilde{\mathbf{F}}_{\rm FD}^{\rm opt}}=\mathbf{U}_{\widetilde{\mathbf{F}}_{\rm FD}^{\rm opt}}\mathbf{\Sigma}_{\widetilde{\mathbf{F}}_{\rm FD}^{\rm opt}}\mathbf{V}_{\widetilde{\mathbf{F}}_{\rm FD}^{\rm opt}}^H$, where $\mathbf{U}_F$ corresponds to the principal components
		\State Extract the phases by using an intermediate variable ${{\bf{F}}_{{\rm{int}}}} = \frac{1}{{\sqrt {{N_t}} }}\exp \left( {j\angle ({{[{{\bf{U}}_{\widetilde{\mathbf{F}}_{\rm FD}^{\rm opt}}}]}_{:,1:N_t^{{\rm{RF}}}}})} \right)$
        \State Quantization by ${{\bf{F}}_{{\rm{RF}}}} = \frac{1}{{\sqrt {{N_t}} }}\exp \left( {j{\textstyle{{2\pi } \over 2^Q}}{\rm{round}}({\textstyle{{2^Q\angle ({{\bf{F}}_{{\rm{int}}}})} \over {2\pi }}})} \right)$
	\end{algorithmic}
\end{algorithm}
Besides, given the practical RF phase shifters with the quantization bit $Q$ \cite{quan}, the phase shifter values can only come from the set $\mathcal{Q}=\{0,\frac{2\pi}{2^Q},\cdots,\frac{2\pi(2^Q-1)}{2^Q}\}$. Therefore, this quantization processing will be performed by searching for the element from $\mathcal{Q}$
according to the minimum Euclidean distance from $\angle([\mathbf{F}_{\rm RF}]_{i,j})$. Finally, how to obtain $\mathbf{F}_{\rm RF}$
is summarized in Algorithm \ref{alg:PCA_AS}.

Note that our PCA-based approach is essentially different from that proposed in \cite{170825}, whose 
RF precoder
is acquired from the eigenvectors of channel covariance matrix $\mathbf{R}_{\rm cov}\!\!=\!\!\frac{1}{K}\!\!\sum_{k=1}^{K}\!\mathbf{H}^H[k]\mathbf{H}[k]$.
By contrast, the RF precoder in our solution is obtained by solving the  principal components or basis for the common column space of channel matrices $\{\mathbf{{H}}[k]\}_{k=1}^K$ at all subcarriers, where the processing is listed in steps 1$\sim$2 of Algorithm 1. Simulation results further confirm
 the better performance of our solution than that proposed in \cite{170825} for hybrid MIMO systems. However, the channel covariance matrix based design in \cite{170825} may have a lower channel estimation overhead than the PCA solution presented here, a detailed analysis of which is beyond the scope of this paper.

\subsection{PCA-Based Hybrid Combiner Design at RX}\label{combiner_FCA}
Based on the designed $\mathbf{F}_{\rm RF}$ and $\{\mathbf{F}_{\rm BB}[k]\}_{k=1}^K$, we further design the hybrid combiner to minimize $\sum\nolimits_{k = 1}^K {||{\bf{x}}[k] - {\bf{r}}[k]||_2^2}$. Specifically, the optimal fully-digital combiner is the minimum mean square error (MMSE) combiner, i.e., ${\mathbf{W}_{\rm FD}^{\rm opt}}^H[k]=\mathbf{W}_{\rm MMSE}^H[k]$, which can be expressed as
\begin{equation}\label{W_opt}
\begin{aligned}
\mathbf{W}_{\rm MMSE}^H[k]
=&\mathbf{F}_{\rm BB}^H[k]\mathbf{F}_{\rm RF}^H\mathbf{H}^H[k]
(\mathbf{H}[k]\mathbf{F}_{\rm RF}
\mathbf{F}_{\rm BB}[k]\mathbf{F}_{\rm BB}^H[k]
\\&\times\mathbf{F}_{\rm RF}^H\mathbf{H}^H[k]+\sigma_n^2\mathbf{I}_{N_r})^{-1}.
\end{aligned}
\end{equation}
Defining the signal at the receive antennas as $\mathbf{y}[k]\in\mathbb{C}^{N_r\times 1}$ ($1\leq k\leq K$), we formulate the combiner design problem as the following optimization problem,
\begin{equation}\label{MSE}
\begin{aligned}
	&{(\mathbf{W}_{\rm RF}^{\rm opt},\{\mathbf{W}_{\rm BB}^{\rm opt}[k]\}_{k=1}^K)}=\min\limits_{\mathbf{W}_{\rm RF},\{\mathbf{W}_{\rm BB}[k]\}_{k=1}^K}\sum\nolimits_{k=1}^{K}\mathbb{E}[||\mathbf{x}[k]
\\&-\mathbf{W}_{\rm BB}^H[k]\mathbf{W}_{\rm RF}^H\mathbf{y}[k]||_2^2]
	\\&\text{s.t. }[\mathbf{W}_{\rm RF}]_{:,i}\in\mathcal{W}_{\rm RF}\text{ for }1\leq i\leq N_r^{\rm RF}.
\end{aligned}
\end{equation}
Note that if the CMC in (\ref{MSE}) is removed, the solution to (\ref{MSE}) is the optimal fully-digital MMSE combiner (\ref{W_opt}). The objective
function in (\ref{MSE}) can be further written as
\begin{equation}\label{inner_MSE}
\begin{aligned}
&\sum\nolimits_{k=1}^{K}\mathbb{E}[||\mathbf{x}[k]-\mathbf{W}_{\rm BB}^H[k]\mathbf{W}_{\rm RF}^H\mathbf{y}[k]||_2^2]
	\\=&\!\sum\nolimits_{k=1}^{K}\!\text{Tr}(\mathbb{E}[\mathbf{x}\![k]\mathbf{x}\!^H\![k]])\!
\\&-\!2\!\sum\nolimits_{k=1}^{K}\!\mathcal{R}\{\!\text{Tr}(\mathbb{E}[\mathbf{x}\![k]\mathbf{y}\!^H\![k]]\mathbf{W}\!_{\rm RF}\!\mathbf{W}\!_{\rm BB}[k])\!\}
\\&+\text{Tr}(\mathbf{W}_{\rm BB}^H[k]\mathbf{W}_{\rm RF}^H\mathbb{E}[\mathbf{y}\![k]\mathbf{y}\!^H\![k]]\mathbf{W}\!_{\rm RF}\!\mathbf{W}\!_{\rm BB}[k]).
\end{aligned}
\end{equation}
Since the optimization variables in (\ref{MSE}) are $\mathbf{W}_{\rm RF}$ and $\{\mathbf{W}_{\rm BB}[k]\}_{k=1}^K$, terms unrelated to $\mathbf{W}_{\rm RF}$ and $\{\mathbf{W}_{\rm BB}[k]\}_{k=1}^K$ will not affect the solution. By adding the independent term $\sum_{k=1}^{K}\!\!\text{Tr}({\mathbf{W}_{\rm FD}^{\rm opt}}^H\![k]\\\times\mathbb{E}[\mathbf{y}\![k]\mathbf{y}^H[k]]{\mathbf{W}_{\rm FD}^{\rm opt}}[k]) -\sum_{k=1}^{K}\text{Tr}(\mathbb{E}[\mathbf{x}[k]\mathbf{x}^H[k]])$ to the objective function (\ref{inner_MSE}), the optimization problem (\ref{MSE}) can be further expressed as
\begin{equation}\label{problem_W}
\begin{aligned}
\min\limits_{\mathbf{W}_{\rm RF},\{\mathbf{W}_{\rm BB}[k]\}_{k=1}^K}&\sum\nolimits_{k=1}^{K}||\mathbb{E}[\mathbf{y}[k]\mathbf{y}^H[k]]^{\frac{1}{2}}({\mathbf{W}_{\rm FD}^{\rm opt}}[k]
\\&-\mathbf{W}_{\rm RF}\mathbf{W}_{\rm BB}[k])||_F^2,
\\&\text{s.t. }[\mathbf{W}_{\rm RF}]_{:,i}\in\mathcal{W}_{\rm RF}\text{ for }1\leq i\leq N_r^{\rm RF},
\end{aligned}
\end{equation}
where $\mathbb{E}[\mathbf{y}[k]\mathbf{y}^H[k]]$
has the closed-form expression of $\mathbf{H}[k]\mathbf{F}_{\rm RF}\mathbf{F}_{\rm BB}[k]\mathbf{F}_{\rm BB}^H[k]\mathbf{F}_{\rm RF}^H\mathbf{H}^H[k]+\sigma_n^2\mathbf{I}_{N_r}$,
and it can be easily calculated. 
For (\ref{problem_W}), it is difficult to jointly optimize $\mathbf{W}_{\rm RF}$ and $\{\mathbf{W}_{\rm BB}[k]\}_{k=1}^K$ due to
the coupling between baseband and RF combiners. Therefore, we will design $\mathbf{W}_{\rm RF}$ and $\{\mathbf{W}_{\rm BB}[k]\}_{k=1}^K$, separately.
Similar to the hybrid precoder design, we first consider the weighted LS estimation of $\mathbf{W}_{\rm BB}[k]$ by fixing $\mathbf{W}_{\rm RF}$ as
\begin{equation}\label{W_BB_LS}
\begin{aligned}
  \mathbf{W}_{\rm BB}[k]=&(\mathbf{W}_{\rm RF}^H\mathbb{E}[\mathbf{y}[k]\mathbf{y}^H[k]]\mathbf{W}_{\rm RF})^{-1}\mathbf{W}_{\rm RF}^H\mathbb{E}[\mathbf{y}[k]\mathbf{y}^H[k]]
  \\&\times{\mathbf{W}_{\rm FD}^{\rm opt}}[k].
\end{aligned}
\end{equation}
In this way, the joint optimization in (\ref{problem_W}) is decoupled. Moreover, similar to the PCA-based hybrid precoder design, we will
 design the RF combiner $\mathbf{W}_{\rm RF}$ from the optimal fully-digital combiner $\{{\mathbf{W}_{\rm FD}^{\rm opt}}[k]\}_{k=1}^K$ based on the weighted PCA. This process is shown in Proposition 2.

\begin{prop}
	Given $\mathbf{W}=\begin{bmatrix}
    \mathbb{E}[\mathbf{y}[1]\mathbf{y}[1]^H]^{1/2}{\mathbf{W}_{\rm FD}^{\rm opt}}[1] \ \cdots \ \mathbb{E}[\mathbf{y}[K]\mathbf{y}[K]^H]^{1/2}{\mathbf{W}_{\rm FD}^{\rm opt}}[K]
	\end{bmatrix}$, the SVD $\mathbf{W}=\mathbf{U}_W\mathbf{\Sigma}_W\mathbf{V}_W^H$, and the weighted LS estimation of $\mathbf{W}_{\rm BB}[k]$ in (\ref{W_BB_LS}), the sub-optimal $\mathbf{W}_{\rm RF}$ to (\ref{problem_W}) is $\mathbf{W}_{\rm RF}=\frac{1}{{\sqrt {{N_r}} }}{\rm exp}(j\angle([\mathbf{U}_W]_{:,1:N_r^{\rm RF}}))$.
\end{prop}
\begin{proof}
	See Appendix C.
\end{proof}

Similar to $\mathbf{F}_{\rm RF}$, we can design $\mathbf{W}_{\rm RF}$ using
Algorithm \ref{alg:PCA_AS} by replacing the input parameters $\{{\mathbf{F}_{\rm FD}^{\rm opt}}[k]\}_{k=1}^K$,
		$N_t^{\rm RF}$,
		and $N_t$ for TX with $\{ [ \mathbf{y}[k]\mathbf{y}[k]^H]^{1/2}{\mathbf{W}_{\rm FD}^{\rm opt}}[k]  \}_{k=1}^K$,
		$N_r^{\rm RF}$,
		and $N_r$ for RX.

\vspace*{-3mm}
\section{Hybrid Precoder/Combiner Design for PCS}
In this section, we first investigate the hybrid precoder/combiner design for FS. Moreover, we study how to group the antennas for AS to further improve the SE performance.
\vspace*{-3mm}
\subsection{Hybrid Precoder/Combiner Design for FS}\label{HDCD}
For FS, each antenna is only connected to one RF chain.
We define the set of antenna indexes as $\{1,\cdots,N_t\}$ and $\mathcal{S}_{l}$ ($1\leq {l}\leq N_t^{\rm RF}$) as the subset of the antennas connected to the ${l}$th RF chain.
Besides, we assume ${\text{card}(\mathcal{S}_{l})}=N_t^{\rm sub}=N_t/N_t^{\rm RF}\in \mathbb{Z}$ and $\mathcal{S}_{l}=\{({l}-1)N_t^{\rm sub}+1,\cdots,{l}N_t^{\rm sub}\}$, $\forall {l}$, for ease of analysis.
%
%
Hence $\mathbf{F}_{\rm RF}=\text{blkdiag}(\mathbf{f}_{{\rm RF},\mathcal{S}_1},\cdots,\mathbf{f}_{{\rm RF},\mathcal{S}_{N_t^{\rm RF}}})$, where $\mathbf{f}_{{\rm RF},\mathcal{S}_{l}}\in\mathbb{C}^{N_t^{\rm sub}\times 1}$ is the analog precoder for the ${l}$th subarray connected to the ${l}$th RF chain.

To design the hybrid precoder for FS,
we first consider the equivalent RF precoder  $\mathbf{\bar{F}}_{\rm RF}=\mathbf{F}_{\rm RF}{(\mathbf{F}_{\rm RF}^H\mathbf{F}_{\rm RF})}^{-1/2}$~\cite{170825},
and it can be further written as
\begin{equation}\label{F_bar}
\mathbf{\bar{F}}_{\rm RF} =\text{blkdiag}(\frac{\mathbf{f}_{{\rm RF},\mathcal{S}_1}}{||\mathbf{f}_{{\rm RF},\mathcal{S}_1}||_2},\cdots,\frac{\mathbf{f}_{{\rm RF},\mathcal{S}_{N_t^{\rm RF}}}}{||\mathbf{f}_{{\rm RF},\mathcal{S}_{N_t^{\rm RF}}}||_2}).
\end{equation}
For the FS, according to Lemma 1, the optimization problem (\ref{opt_final}) can be simplified as
\begin{equation}\label{opt_final_eq}
\begin{aligned}
\max_{\mathbf{F}_{\rm RF}}&\sum\nolimits_{k=1}^{K}||{\mathbf{F}_{\rm FD}^{\rm opt}}^H[k]\mathbf{\bar{F}}_{\rm RF}||_F^2,
\\&\text{s.t. } \left( {\ref{F_bar}} \right),~
\mathbf{f}_{{\rm RF},\mathcal{S}_{l}}\in\mathcal{F}_{{\rm RF},\mathcal{S}},\forall {l},
\end{aligned}
\end{equation}
where $\mathcal{F}_{{\rm RF},\mathcal{S}}$ is a set of feasible RF precoder satisfying the CMC, and the optimal fully-digital precoder ${\mathbf{F}_{\rm FD}^{\rm opt}}[k]=[\mathbf{V}[k]]_{:,1:N_s}\in\mathbb{C}^{N_t\times N_s}$ can be expressed
as the following block matrix form
\begin{equation}\label{blk_F_opt}
{\mathbf{F}_{\rm FD}^{\rm opt}}^H[k]=\begin{bmatrix}
\mathbf{F}_{{\rm opt},\mathcal{S}_1}^H[k] \
\cdots \
\mathbf{F}_{{\rm opt},\mathcal{S}_{N_t^{\rm RF}}}^H[k]
\end{bmatrix}.
\end{equation}
Here $\mathbf{F}_{{\rm opt},\mathcal{S}_{l}}[k]\in\mathbb{C}^{N_t^{\rm sub}\times N_s}$, $\forall k$. To solve (\ref{opt_final_eq}), we use the following proposition.
\begin{prop}
    For FS, given $\mathbf{F}_{\mathcal{S}_{l}}=\begin{bmatrix}
	\mathbf{F}_{{\rm opt},\mathcal{S}_{l}}[1] &\cdots& \mathbf{F}_{{\rm opt},\mathcal{S}_{l}}[K]\end{bmatrix}$, the sub-optimal $\mathbf{F}_{\rm RF}$ to (\ref{opt_final_eq}) is $\mathbf{F}_{\rm RF}=\text{\rm blkdiag}(\mathbf{f}_{{\rm RF},\mathcal{S}_1},\cdots,\mathbf{f}_{{\rm RF},\mathcal{S}_{N_t^{\rm RF}}})$, where $\mathbf{f}_{{\rm RF},\mathcal{S}_l}=\frac{1}{{\sqrt {{N_t^{\rm sub}}} }}{\rm exp}(j\angle(\mathbf{u}_{\mathcal{S}_l,1}))$, $\mathbf{u}_{\mathcal{S}_{l},1}\in\mathbb{C}^{N_t^{\rm sub}\times 1}$ is the right singular vector of the largest singular value of $\mathbf{F}_{\mathcal{S}_{l}}$, $\forall {l}$.
\end{prop}
\begin{proof}
	By substituting (\ref{F_bar}) and (\ref{blk_F_opt}) into the objective function of (\ref{opt_final_eq}), we can further have
\begin{equation}\label{r_F_PCS}
	\begin{aligned}
    &\sum_{k=1}^{K}||{\mathbf{F}_{\rm FD}^{\rm opt}}^H[k]\mathbf{\bar{F}}_{\rm RF}||_F^2
    \\=&\sum\nolimits_{k=1}^{K}||\begin{bmatrix}
	\frac{\mathbf{F}_{{\rm opt},\mathcal{S}_1}^H[k]\mathbf{f}_{{\rm RF},\mathcal{S}_1}}{||\mathbf{f}_{{\rm RF},\mathcal{S}_1}||_2} &
	\cdots &
	\frac{\mathbf{F}_{{\rm opt},\mathcal{S}_{N_t^{\rm RF}}}^H[k]\mathbf{f}_{{\rm RF},\mathcal{S}_{N_t^{\rm RF}}}}{||\mathbf{f}_{{\rm RF},\mathcal{S}_{N_t^{\rm RF}}}||_2}
	\end{bmatrix}||_F^2
    \\=&\sum\nolimits_{{l}=1}^{N_t^{\rm RF}}\frac{\sum_{k=1}^{K}||\mathbf{F}_{{\rm opt},\mathcal{S}_{l}}^H[k]\mathbf{f}_{{\rm RF},\mathcal{S}_{l}}||_2^2}{||\mathbf{f}_{{\rm RF},\mathcal{S}_{l}}||_2^2}
    \\=&\sum\nolimits_{{l}=1}^{N_t^{\rm RF}}\frac{\mathbf{f}_{{\rm RF},\mathcal{S}_{l}}^H\mathbf{F}_{\mathcal{S}_{l}}\mathbf{F}_{\mathcal{S}_{l}}^H\mathbf{f}_{{\rm RF},\mathcal{S}_{l}}}{||\mathbf{f}_{{\rm RF},\mathcal{S}_{l}}||_2^2}\leq\sum\nolimits_{{l}=1}^{N_t^{\rm RF}}\lambda_1^2(\mathbf{F}_{\mathcal{S}_{l}}),
	\end{aligned}
	\end{equation}
where $\lambda_1(\mathbf{F}_{\mathcal{S}_{l}})$ means the largest singular value of matrix $\mathbf{F}_{\mathcal{S}_{l}}$. The maximum value can only be obtained if $\mathbf{f}_{{\rm RF},\mathcal{S}_l}=\frac{1}{{\sqrt {{N_t^{\rm sub}}} }}{\rm exp}(j\angle(\mathbf{u}_{\mathcal{S}_l,1}))$, where $\mathbf{u}_{\mathcal{S}_{l},1}$ is the left singular vector of the largest singular value of the matrix $\mathbf{F}_{\mathcal{S}_r}$.
\end{proof}

By taking quantization of phase shifters into account, the final $\mathbf{F}_{\rm RF}$ is
${\bf{f}}_{{\rm{RF}},{{\cal S}_{l}}}= \frac{1}{{\sqrt {{N_t^{\rm sub}}} }}\\\times{\rm exp}(j{\textstyle{2\pi\over {2^Q}}}
{\rm{round}}({\textstyle{{2^Q\!\angle ({{\bf{u}}_{{{\cal S}_{l}},1}})} \over {2\pi }}}))$. Meanwhile,
$\left\{ {{{\bf{F}}_{{\rm{BB}}}}[k]} \right\}_{k = 1}^K$ can be obtained according to~(\ref{F_BB}).

At the RX, we consider the subset of antenna indices connected to the ${l}$th RF chain as $\mathcal{T}_{l}=\{({l}-1)N_r^{\rm sub}+1,\cdots,{l}N_r^{\rm sub}\}$, $\forall {l}$, where ${\text{card}(\mathcal{T}_{l})}=N_r^{\rm sub}=N_r/N_r^{\rm RF}\in \mathbb{Z}$ for ease of analysis. Similar to the TX, we have $\mathbf{W}_{\rm RF}=\text{blkdiag}(\mathbf{w}_{{\rm RF},\mathcal{T}_1},\cdots,\mathbf{w}_{{\rm RF},\mathcal{T}_{N_r^{\rm RF}}})\in\mathbb{C}^{N_r\times N_r^{\rm RF}}$ and $\mathbf{W}_{\rm BB}[k]=[
\mathbf{w}_{{\rm BB},\mathcal{T}_1}[k]
\ \cdots
\ \mathbf{w}_{{\rm BB},\mathcal{T}_{N_r^{\rm RF}}}[k]]^H\in\mathbb{C}^{N_r^{\rm RF}\times N_s}$,
where $\mathbf{w}_{{\rm BB},\mathcal{T}_{l}}[k]\in\mathbb{C}^{N_s\times 1}$ is the baseband combiner of the ${l}$th receive subarray. Hence 
$      \mathbf{W}_{\rm RF}\mathbf{W}_{\rm BB}[k]\!\!=\!\!
[\mathbf{w}_{{\rm RF},\mathcal{T}_1}\mathbf{w}_{{\rm BB},\mathcal{T}_1}^H\![k]
\ \cdots
\ \mathbf{w}_{{\rm RF},\mathcal{T}_{N_r^{\rm RF}}}\mathbf{w}_{{\rm BB},\mathcal{T}_{N_r^{\rm RF}}}^H\![k]]^T\in\mathbb{C}^{N_r\times N_s}$.
Moreover, consider the effective channel
\begin{equation}\label{H_eff}
  \mathbf{H}_{\rm eff}[k]=\mathbf{H}[k]\mathbf{F}_{\rm RF}\mathbf{F}_{\rm BB}[k],
\end{equation}
the received signal $\mathbf{y}[k]=[(\mathbf{y}_{\mathcal{T}_1}[k])^T\ \cdots \ (\mathbf{y}_{\mathcal{T}_{N_r^{\rm RF}}}[k])^T]^T\in\mathbb{C}^{N_r\times 1}$ with
$\mathbf{y}_{\mathcal{T}_{l}}[k]\in\mathbb{C}^{N_r^{\rm sub}\times 1}$, the effective channel
$\mathbf{H}_{\rm eff}[k]=[(\mathbf{H}_{\rm eff,\mathcal{T}_1}[k])^T\ \cdots \ (\mathbf{H}_{{\rm eff},\mathcal{T}_{N_r^{\rm RF}}}[k])^T]^T\in\mathbb{C}^{N_r\times N_s}$ with
$\mathbf{H}_{\rm eff,\mathcal{T}_{l}}[k]\in\mathbb{C}^{N_r^{\rm sub}\times N_s}$, and the noise $\mathbf{n}[k]=[(\mathbf{n}_{\mathcal{T}_1}[k])^T\ \cdots \ (\mathbf{n}_{\mathcal{T}_{N_r^{\rm RF}}}[k])^T]^T\in\mathbb{C}^{N_r\times 1}$ with $\mathbf{n}_{\mathcal{T}_{l}}[k]\in\mathbb{C}^{N_r^{\rm sub}\times 1}$, we have
\begin{equation}\label{blk_y}
{{\bf{y}}_{{{\cal T}_{l}}}}[k]{\rm{ = }}{{\bf{H}}_{{\rm{eff}},{{\cal T}_{l}}}}[k]{\bf{x}}[k] + {{\bf{n}}_{{{\cal T}_{l}}}}[k], \forall {l}.
\end{equation}
By substituting (\ref{blk_y}) into (\ref{MSE}), we have
\begin{equation}\label{MSE_2}
\begin{aligned}
&\sum_{k=1}^{K}\!\mathbb{E}[||\mathbf{x}\![k]\!\!-\!\!\mathbf{W}_{\rm BB}^H\![k]\mathbf{W}_{\rm RF}^H\mathbf{y}\![k]||_2^2]
\\\!=\!\!&\sum_{k=1}^{K}\!(\text{Tr}(\mathbb{E}[\mathbf{x}\![k]\mathbf{x}^H\![k]]\!)
\\\!&-\!\!2\!\!\sum_{{l}=1}^{N_r^{\rm RF}}\!\mathfrak{R}\{\text{Tr}(\mathbb{E}[\mathbf{x}\![k]\mathbf{y}_{\mathcal{T}_{l}}^H\![k]]\mathbf{w}_{{\rm RF},\mathcal{T}_{l}}\!\mathbf{w}_{{\rm BB},\mathcal{T}_{l}}^H\![k]\!)\!\}
\\&+\sum_{{l}=1}^{N_r^{\rm RF}}\text{Tr}(\mathbf{w}_{{\rm BB},\mathcal{T}_{l}}[k]\mathbf{w}_{{\rm RF},\mathcal{T}_{l}}^H\mathbb{E}[\mathbf{y}_{\mathcal{T}_{l}}[k]
\mathbf{y}_{\mathcal{T}_{l}}^H[k]]\mathbf{w}_{{\rm RF},\mathcal{T}_{l}}
\mathbf{w}_{{\rm BB},\mathcal{T}_{l}}^H[k])),
\end{aligned}
\end{equation}
where $\mathbb{E}[\mathbf{y}_{\mathcal{T}_{l}}[k]\mathbf{y}_{\mathcal{T}_{l}}^H[k]]=
\mathbf{H}_{{\rm eff},\mathcal{T}_{l}}[k]\mathbf{H}_{{\rm eff},\mathcal{T}_{l}}^H[k]+\sigma_n^2\mathbf{I}_{N_r^{\rm sub}}$.
To design the $\{\mathbf{W}_{\rm BB}[k]\}_{k=1}^K$ and $\mathbf{W}_{\rm RF}$ for minimizing (\ref{MSE_2}),
we first use the MMSE criterion to obtain the optimal fully-digital combiner according to (\ref{W_opt}), which can be expressed as
${\mathbf{W}_{\rm FD}^{\rm opt}}^H[k]=\begin{bmatrix}
\mathbf{W}_{{\rm opt},\mathcal{T}_1}^H[k] \
\cdots \
\mathbf{W}_{{\rm opt},\mathcal{T}_{N_r^{\rm RF}}}^H[k]
\end{bmatrix}\in\mathbb{C}^{N_s\times N_r}$ 
 with $\mathbf{W}_{{\rm opt},\mathcal{T}_{l}}[k]\in\mathbb{C}^{N_r^{\rm sub}\times N_s}$, $\forall k$. Furthermore, we can transform (\ref{MSE_2}) into equation (\ref{min_f}) below by adding a constant term $\sum\limits_{k = 1}^K {\sum\limits_{{l} = 1}^{N_r^{{\rm{RF}}}} {{\rm{Tr}}} } ({\bf{W}}_{{\rm{opt}},{{\cal T}_{l}}}^H[k]\mathbb{E}[{{\bf{y}}_{{{\cal T}_{l}}}}[k]{\bf{y}}_{{{\cal T}_{l}}}^H[k]] {{\bf{W}}_{{\rm{opt}},{{\cal T}_{l}}}}[k]) - \sum\limits_{k = 1}^K {{\rm{Tr}}} (\mathbb{E}[{\bf{x}}[k]{{\bf{x}}^H}[k]])$ irrelevant to the optimization object variables $\{\mathbf{W}_{\rm BB}[k]\}_{k=1}^K$ and $\mathbf{W}_{\rm RF}$
\begin{equation}\label{min_f}
\sum_{k=1}^{K}\sum_{{l}=1}^{N_r^{\rm RF}}||\mathbb{E}[\mathbf{y}_{\mathcal{T}_{l}}[k]
\mathbf{y}_{\mathcal{T}_{l}}^H[k]]^{\frac{1}{2}}(\mathbf{W}_{{\rm opt},\mathcal{T}_{l}}[k]
-\mathbf{w}_{{\rm RF},\mathcal{T}_{l}}\mathbf{w}_{{\rm BB},\mathcal{T}_{l}}^H[k])||_F^2.
\end{equation}
To minimize (\ref{min_f}), we consider
the weighted LS estimation of $\mathbf{W}_{\rm BB}[k]$, denoted by $\mathbf{W}^{\rm WLS}_{\rm BB}[k]$, according to
(\ref{W_BB_LS}), then we
have a sub-optimal $\mathbf{w}_{{\rm RF},\mathcal{T}_l}=\frac{1}{{\sqrt {{N_r^{\rm sub}}} }}{\rm exp}(j\angle(\mathbf{u}_{\mathcal{T}_l,1}))$,
where $\mathbf{u}_{\mathcal{T}_{l},1}$ is the left singular vector of the largest singular value of $\mathbf{W}_{\mathcal{T}_{l}}=[ \mathbb{E}[\mathbf{y}_{\mathcal{T}_{l}}[1]
\mathbf{y}_{\mathcal{T}_{l}}^H[1]]^{\frac{1}{2}}\mathbf{W}_{{\rm opt},\mathcal{T}_{l}}[1] \cdots\ \mathbb{E}[\mathbf{y}_{\mathcal{T}_{l}}[K]
\mathbf{y}_{\mathcal{T}_{l}}^H[K]]^{\frac{1}{2}} \mathbf{W}_{{\rm opt},\mathcal{T}_{l}}[K]]$. At last, $\{\mathbf{W}_{\rm BB}[k]\}_{k=1}^K$ can be obtained by (\ref{W_BB_LS}). 
\vspace*{-3mm}
\subsection{Antenna Grouping for Hybrid Precoder in AS}

For AS, how to group the transmit/receive antennas, i.e., design $\{\mathcal{S}_{l}\}_{{l}=1}^{N_t^{\rm RF}}$ and $\{\mathcal{T}_{l}\}_{{l}=1}^{N_r^{\rm RF}}$ can further improve the SE performance.
At the TX, the optimization of transmit antenna grouping $\{\mathcal{S}_{l}\}_{{l}=1}^{N_t^{\rm RF}}$ can be formulated as the following optimization problem according to (\ref{r_F_PCS})
\begin{equation}\label{pro_dy}
\setlength{\abovedisplayskip}{12pt}
\setlength{\belowdisplayskip}{12pt}
\begin{aligned}
\max\limits_{\mathcal{S}_1,\cdots,\mathcal{S}_{N_t^{\rm RF}}}&\sum\nolimits_{{l}=1}^{N_t^{\rm RF}}\lambda_1^2(\mathbf{F}_{\mathcal{S}_{l}})
\\&\text{s.t. }\cup_{{l}=1}^{N_t^{\rm RF}}\mathcal{S}_{l}=\{1,\cdots,N_t\},
\\&\mathcal{S}_i\cap\mathcal{S}_j=\emptyset\text{ for }i\not=j,\ \mathcal{S}_{l}\not=\emptyset\ \forall {l}.
\end{aligned}
\end{equation}
This optimization problem is a combinational optimization problem, which requires an exhaustive search to reach the optimal solution. The number of all possible combinations to obtain the optimal solution is $\frac{1}{(N_t^{\rm RF})!}\sum_{n=0}^{N_r^{\rm RF}}(-1)^{N_t^{\rm RF}-n}\binom{N_t^{\rm RF}}{n}n^{N_t}$ according to \cite{S_num}. For example, when $N_t=64$ and $N_t^{\rm RF}=4$, the number of all possible combinations can be up to $1.4178\times 10^{37}$.
Therefore, we will design a low-complexity antenna grouping algorithm to maximize (\ref{pro_dy}). Specifically, given $\mathbf{R}_{\mathcal{S}_{l}}=\mathbf{F}_{\mathcal{S}_{l}}\mathbf{F}_{\mathcal{S}_{l}}^H$ and $\mathbf{R}_F=\mathbf{F}\mathbf{F}^H$, we have $\lambda_1^2(\mathbf{F}_{\mathcal{S}_{l}})=\lambda_1(\mathbf{R}_{\mathcal{S}_{l}})$ and the following approximation
\begin{equation}\label{lambda_approx}
\setlength{\abovedisplayskip}{12pt}
\setlength{\belowdisplayskip}{12pt}
\begin{aligned}
\lambda_1(\mathbf{R}_{\mathcal{S}_{l}})&\approx \frac{1}{\text{card}(\mathcal{S}_{l})}\sum\nolimits_{i=1}^{\text{card}(\mathcal{S}_{l})}\sum\nolimits_{j=1}^{\text{card}(\mathcal{S}_{l})}|[\mathbf{R}_{\mathcal{S}_{l}}]_{i,j}|
\\&=\frac{1}{\text{card}(\mathcal{S}_{l})}\sum_{i\in\mathcal{S}_r }\sum_{j\in\mathcal{S}_{l}}|[\mathbf{R}_F]_{i,j}|,
\end{aligned}
\end{equation}
which is due to the tight lower bound and upper bound of $\lambda_1(\mathbf{R}_{\mathcal{S}_{l}})$ as proven in \cite{170825}.
Hence, the objective function of (\ref{pro_dy}) becomes $\sum_{{l}=1}^{N_t^{\rm RF}}\frac{1}{\text{card}(\mathcal{S}_{l})}\sum_{i\in\mathcal{S}_{l} }\sum_{j\in\mathcal{S}_{l}}|[\mathbf{R}_F]_{i,j}|$.
This is still a combinational optimization problem, which requires the exhaustive search with high complexity.

\begin{algorithm}[!tb]
	\caption{Proposed Shared-AHC Algorithm to Group Antennas for AS.}
	\label{alg:AHC}
	\begin{algorithmic}[1]
		\renewcommand{\algorithmicrequire}{\textbf{Input:}}
		\renewcommand\algorithmicensure {\textbf{Output:}}
		\Require
		The correlation matrix $\mathbf{R}_F=\mathbf{F}\mathbf{F}^H$, number of antennas $N_t$, number of RF chains $N_t^{\rm RF}$.
		\Ensure
		Antenna grouping results $\mathcal{S}_1,\cdots,\mathcal{S}_{N_t^{\rm RF}}$.
		\State $N_{\rm sub}=N_t$, $\mathcal{S}_i=\{i\}$ for $i=1,\cdots,N_t$
		\While{$N_{\rm sub}>N_t^{\rm RF}$}
		\State $\mathcal{S}_i^0=\mathcal{S}_i$ for $i=1,\cdots,N_{\rm sub}$, $n_{\rm sub}=1$
		\For{$i=1:N_{\rm sub}$}
		\State \textbf{if} $\exists r_0\ne i\text{ s.t. }\mathcal{S}_i^0\subseteq\mathcal{S}_{r_0}$ \textbf{then continue}
		\State \textbf{else if} $i=N_{\rm sub}$ \textbf{then} $\mathcal{S}_{n_{\rm sub}}=\mathcal{S}_i^0$
		\State \textbf{else} $j=\arg\max\limits_{l\in\{i+1,\cdots,N_{\rm sub}\}}g(\mathcal{S}_i,\mathcal{S}_l)$, $i^0=\arg\max\limits_{l\in\{1,\cdots,N_{\rm sub}\}\setminus\{j\}}g(\mathcal{S}_j,\mathcal{S}_l)$
		\State \qquad \textbf{if} $i=i^0$ \textbf{then} $\mathcal{S}_{n_{\rm sub}}=\mathcal{S}_i\cup\mathcal{S}_j$
		\State \qquad \textbf{else} $\mathcal{S}_{n_{\rm sub}}=\mathcal{S}_i$
		\State \qquad\textbf{end if}
		\State \textbf{end if}
		\State $n_{\rm sub}=n_{\rm sub}+1$
		\EndFor
		\State $N_{\rm sub}^0=n_{\rm sub}-1$
		\State \textbf{if} {$N_{\rm sub}^0<N_r^{\rm RF}$} \textbf{then} $\mathcal{S}_i=\mathcal{S}_i^0$ for $i=1,\cdots,N_{\rm sub}$ \textbf{break}
		\State \textbf{else} $N_{\rm sub}=N_{\rm sub}^0$
		\State \textbf{end if}
		\EndWhile
		\State \textbf{if} $N_{\rm sub}>N_t^{\rm RF}$ \textbf{then} sort $\mathcal{S}_i$ according to the ascending order of cardinality
		\State \qquad \textbf{for} $i=1:(N_{\rm sub}-N_t^{\rm RF})$ \textbf{do} $j=\arg\max\limits_{l=\{N_t^{\rm RF}-N_{\rm sub}+1,\cdots,N_{\rm sub}\}}g(\mathcal{S}_i,\mathcal{S}_l)$, $\mathcal{S}_i=\mathcal{S}_i\cup\mathcal{S}_j$
		\State \qquad \textbf{end for}
		\State \qquad Rearrange the subscript to guarantee that the order of subscripts are from 1 to $N_t^{\rm RF}$
		\State \textbf{end if}
	\end{algorithmic}
\end{algorithm}

In this paper, we formulate the antenna grouping problem as the clustering analysis problem in machine learning.
Since $\mathbf{R}_F$ is a correlation metric rather than the distance metric, we focus on the correlation-based clustering approach
 and propose the shared-AHC algorithm as listed in Algorithm \ref{alg:AHC}.
 The proposed algorithm is developed from the AHC algorithm \cite{AHC}, and it can divide the antennas into multiple groups connected to different RF chains.
Note that the traditional AHC algorithm builds a cluster hierarchy from the bottom up, and it starts by adding all data to multiple clusters, followed by iteratively pair-wise merging these clusters until only one cluster is left at the top of the hierarchy \cite{AHC}.
By contrast, the proposed shared-AHC algorithm simultaneously builds $N_t^{\rm RF}$ clusters, rather than only one cluster in conventional AHC algorithm. Besides, the pair-wise merging criterion in the proposed algorithm is ``shared", while the conventional AHC algorithm only focuses on the target cluster. To further illustrate this ``shared" mechanism, we  introduce the metric of mutual correlation $g(\mathcal{S}_n,\mathcal{S}_m)$ for any two clusters $\mathcal{S}_n$ and $\mathcal{S}_m$ as
\begin{equation}\label{mul_cor}
g(\mathcal{S}_n,\mathcal{S}_m)=\frac{1}{\text{card}(\mathcal{S}_n)\text{card}(\mathcal{S}_m)}
\sum_{i\in\mathcal{S}_n}\sum_{j\in\mathcal{S}_m}|[\mathbf{R}_F]_{i,j}|,m \ne n.
\end{equation}
In each clustering iteration, we first focus on the cluster $\mathcal{S}_n$ and search for a cluster $\mathcal{S}_m$ that maximizes $g(\mathcal{S}_n,\mathcal{S}_l)$ among all possible $\mathcal{S}_l$.
If the cluster $\mathcal{S}_n$ also maximizes $g(\mathcal{S}_m,\mathcal{S}_l)$ among all possible $\mathcal{S}_l$,
we merge $\mathcal{S}_n$ and $\mathcal{S}_m$.
Otherwise, the cluster $\mathcal{S}_n$ and cluster $\mathcal{S}_m$ are not merged, and algorithm goes into the next iteration.
Therefore, our proposed algorithm is featured as ``shared", since two clusters mutually share the maximum correlation in the sense of (\ref{mul_cor}).

The steps of the proposed shared-AHC algorithm are elaborated as follows. Step 1 performs the initialization. 
Step 3 saves the clustering results in the last iteration and initializes the
clustering process counter $n_{\rm sub}$. Step 5 considers the special situation that the target cluster is already
merged into a former cluster, and step 6 considers the situation that the target cluster is $\mathcal{S}_{N_{\rm sub}}$ and not merged into any of the former clusters.
To
maximize the correlation $g(\mathcal{S}_i,\mathcal{S}_l)$ for the target cluster $\mathcal{S}_i$,
the operation $j=\arg\max\limits_{l\in\{i+1,\cdots,N_{\rm sub}\}}g(\mathcal{S}_i,\mathcal{S}_l)$ in step 7
will search for the cluster $\mathcal{S}_j$
from the clusters that have not been searched, i.e., $\left\{ {{{\cal S}_j}} \right\}_{l = {\rm{i + 1}}}^{{N_{{\rm{sub}}}}}$.
The operation $i^0=\arg\max\limits_{l\in\{1,\cdots,N_{\rm sub}\}\setminus\{j\}}g(\mathcal{S}_j,\mathcal{S}_l)$ in step 7 further judges whether
$\mathcal{S}_i$ also has the maximum mutual correlation for the chosen $\mathcal{S}_j$ or not. 
$i=i^0$ in step 8 indicates this judge holds, then clusters $\mathcal{S}_i$ and
$\mathcal{S}_j$ are merged. Otherwise, these two clusters will not be merged (step 9). After processing $\mathcal{S}_i$, the
cluster counter $n_{\rm sub}$ is increased by one, and the next target cluster will be processed (step 12).
After the loop including steps 4-13 finishes, 
if the resulting number of clusters $N_{\rm sub}<N_r^{\rm RF}$ (step 15), the iteration stops and the
clustering result of the last iteration will be considered, i.e.,
$\mathcal{S}_i=\mathcal{S}_i^0$ for $i=1,\cdots,N_{\rm sub}$.
Otherwise, we continue the iteration
(step 16). Steps 19-23 ensure the result $N_{\rm sub}=N_t^{\rm RF}$. If $N_{\rm sub}>N_t^{\rm RF}$,
the $(N_{\rm sub}-N_t^{\rm RF})$ clusters with $(N_{\rm sub}-N_t^{\rm RF})$ smallest
cardinality are merged within the rest $N_t^{\rm RF}$ clusters (steps 19-21). Step 22
guarantees that the subscripts of clustering result match the notation of RF chains.

Note that although the antenna grouping $\{\mathcal{S}_{l}\}_{{l}=1}^{N_t^{\rm RF}}$ designed by
Algorithm \ref{alg:AHC} is based on the instantaneous CSI
$\mathbf{H}[k]$, $\{\mathcal{S}_{l}\}_{{l}=1}^{N_t^{\rm RF}}$ mainly depends on
the steering vectors having the first $N_s$ largest path gains for massive MIMO with large $N_t$ (Proposition 4).
On the other hand, for time-varying MIMO channels, the variation rates
for channel angles and the absolute values of channel gains are usually much
slower than that for channel gains \cite{tvchannel}.
%
This indicates that once
$\{\mathcal{S}_{l}\}_{{l}=1}^{N_t^{\rm RF}}$ is determined, it can remain unchangeable
for a long period of time with negligible performance loss.
The proof for Proposition 4 is provided as follows.
\begin{prop}\label{change}
    For massive MIMO with large $N_t$, the correlation matric $\mathbf{R}_F$ for antenna grouping 
     only depends on the steering vectors associated with the first $N_s$ largest path gains.
\end{prop}
\begin{proof}
	For massive MIMO, the transmit/receive steering vectors are asymptotic orthogonal, i.e., $\lim\limits_{N_t\to \infty}\mathbf{A}_t^H\mathbf{A}_t=\mathbf{I}_{N_{\rm cl}N_{\rm ray}}$ and $\lim\limits_{N_r\to \infty}\mathbf{A}_r^H\mathbf{A}_r=\mathbf{I}_{N_{\rm cl}N_{\rm ray}}$~\cite{mao}.
Furthermore, we assume an ideal pulse-shaping $p(t)=\delta(t)$ and $|{\alpha _{1,1}}| > |{\alpha _{1,2}}| >  \cdots  > |{\alpha _{{N_{{\rm{cl}}}},{N_{{\rm{ray}}}}}}|$ for ease of analysis.
So the SVD of $\mathbf{H}[k]$ in (\ref{SVD}) can be written as
$\mathbf{U}[k]\!\!=\!\!\begin{bmatrix}{\mathbf{A}}_r & \mathbf{U}_r\end{bmatrix}$,
                   $\mathbf{\Sigma}[k]\!\!=\!\!\text{blkdiag}(|\mathbf{P}[k]|,\mathbf{0}_{(N_r\!-\!N_{\rm cl}N_{\rm ray})\times(N_t\!-\!N_{\rm cl}N_{\rm ray})})$, and
                   $\mathbf{V}[k]=\text{blkdiag}(e^{j\angle\mathbf{P}[k]},\mathbf{I}_{N_t-N_{\rm cl}N_{\rm ray}})\begin{bmatrix}
                                                               {\mathbf{A}}_t & \mathbf{V}_t
                                                             \end{bmatrix}$,
where $\mathbf{U}_r\in\mathbb{C}^{N_r\times(N_r-N_{\rm cl}N_{\rm ray})}$ and
$\mathbf{V}_t\in\mathbb{C}^{N_t\times(N_t-N_{\rm cl}N_{\rm ray})}$
are semi-unitary matrices respectively satisfying $Col\mathbf{U}_t\!\!=\!\!(Col\mathbf{A}_t)^\bot$ and $Col\mathbf{U}_t\!=(Col\mathbf{A}_t)^\bot$,
${\bf{P}}[k]\!\!=\!\!\!\sqrt {\frac{{{N_t}{N_r}}}{{{N_{{\rm{cl}}}}{N_{{\rm{ray}}}}}}} {\bf{G\tilde P}}[k]$,
${\bf{\tilde P}}[k]\!\!=\!{\rm{diag}}({e^{ - j2\pi {\tau _{1,1}}k\!/\!KT_s}},{e^{ - j2\pi{\tau _{1,2}}k\!/\!KT_s}}, \cdots ,{e^{ - j2\pi{\tau _{{N_{{\rm{cl}}}},{N_{{\rm{ray}}}}}}k\!/\!KT_s}})$, and
${\bf{G}} = {\rm{diag}}({\alpha _{1,1}},{\alpha _{1,2}}, \cdots ,{\alpha _{{N_{{\rm{cl}}}},{N_{{\rm{ray}}}}}})$.
Define the matrix consisting of the steering vectors associated with the first $N_s$ largest singular values as
$\mathbf{A}_{t,{\rm max\_}{N_s}}=[\tilde{\mathbf{A}}_t]_{:,1:N_s}$, we have ${\mathbf{F}_{\rm FD}^{\rm opt}}[1]=\cdots={\mathbf{F}_{\rm FD}^{\rm opt}}[K]=\mathbf{A}_{t,{\rm max\_}{N_s}}$.
Therefore, $\mathbf{R}_F=\mathbf{F}\mathbf{F}^H=K\mathbf{A}_{t,{\rm max}}\mathbf{A}_{t,{\rm max}}^H$, that is to say,
$\mathbf{R}_F$ only depends on the steering vectors associated with the first $N_s$ largest path gains.
\end{proof}
\vspace*{-5mm}
\subsection{Antenna Grouping for AS on Hybrid Combiner}
Furthermore, we consider the antenna grouping for AS at the RX. To decouple $\mathbf{w}_{{\rm RF},\mathcal{T}_{l}}$ and $\{\!\mathbf{w}_{{\rm BB},\mathcal{T}_{l}}\![k]\!\}_{k=1}^K$ in (\ref{min_f}), we rewrite the problem (\ref{MSE}) as
\begin{equation}\label{MSE_32}
\begin{aligned}
	\min\limits_{{{{\bf{w}}_{{\rm{RF}},{{\cal T}_{l}}}},{{\bf{w}}_{{\rm{BB}},{{\cal T}_{l}}}}[k],\forall {l},k}}&\sum\nolimits_{k=1}^{K}\sum\nolimits_{{l}=1}^{N_r^{\rm RF}}||\mathbb{E}[\mathbf{y}_{\mathcal{T}_{l}}[k]
\mathbf{y}_{\mathcal{T}_{l}}^H[k]]^{\frac{1}{2}}
\\&(\mathbf{W}_{{\rm opt},\mathcal{T}_{l}}[k]
-\mathbf{w}_{{\rm RF},\mathcal{T}_{l}}\mathbf{w}_{{\rm BB},\mathcal{T}_{l}}^H[k])||_F^2
	\\&\text{s.t. }\mathbf{w}_{{\rm RF},\mathcal{T}_{l}}\in\mathcal{W}_{{\rm RF},\mathcal{T}_{l}},
\end{aligned}
\end{equation}
where $\mathcal{W}_{{\rm RF},\mathcal{T}_{l}}$ is a set of feasible RF combiner satisfying the CMC. Given the RF combiner $\mathbf{w}_{{\rm RF},\mathcal{T}_{l}}$, the objective function (\ref{MSE_32}) can be rewritten as
\begin{equation}\label{MSE_33}
\begin{aligned}
	\min\limits_{\{\mathbf{w}_{{\rm BB},\mathcal{T}_{l}}[k]\}_{k=1}^K}&\sum\nolimits_{k=1}^{K}||\mathbb{E}[\mathbf{y}_{\mathcal{T}_{l}}[k]
\mathbf{y}_{\mathcal{T}_{l}}^H[k]]^{\frac{1}{2}}(\mathbf{W}_{{\rm opt},\mathcal{T}_{l}}[k]
\\&-\mathbf{w}_{{\rm RF},\mathcal{T}_{l}}\mathbf{w}_{{\rm BB},\mathcal{T}_{l}}^H[k])||_F^2,
\end{aligned}
\end{equation}
for $1\leq {l}\leq N_r^{\rm RF}$. For (\ref{MSE_33}), the optimal baseband combiner can be obtained by weighted LS as
\begin{equation}\label{w_bb_row}
\begin{aligned}
\mathbf{w}_{{\rm BB},\mathcal{T}_{l}}^H[k]=&(\mathbf{w}_{{\rm RF},\mathcal{T}_{l}}^H\mathbb{E}[\mathbf{y}_{\mathcal{T}_{l}}[k]
\mathbf{y}_{\mathcal{T}_{l}}^H[k]]\mathbf{w}_{{\rm RF},\mathcal{T}_{l}})^{-1}
\mathbf{w}_{{\rm RF},\mathcal{T}_{l}}^H
\\\times&\mathbb{E}[\mathbf{y}_{\mathcal{T}_{l}}[k]
\mathbf{y}_{\mathcal{T}_{l}}^H[k]]\mathbf{W}_{{\rm opt},\mathcal{T}_{l}}[k].
\end{aligned}
\end{equation}
By substituting (\ref{w_bb_row}) into (\ref{min_f}), we obtain
\begin{equation}\label{min_dy}
\begin{aligned}
\sum_{k=1}^{K}&\!\sum_{{l}=1}^{N_r^{\rm RF}}\!\text{Tr}(\mathbf{W}_{{\rm opt},\mathcal{T}_{l}}^H[k]\mathbb{E}[\mathbf{y}_{\mathcal{T}_{l}}[k]
\mathbf{y}_{\mathcal{T}_{l}}^H[k]]\mathbf{W}_{{\rm opt},\mathcal{T}_{l}}[k]
\\-&\!\!\mathbf{W}_{{\rm opt},\mathcal{T}_{l}}^H\![k]\mathbb{E}[\mathbf{y}_{\mathcal{T}_{l}}\![k]
\mathbf{y}_{\mathcal{T}_{l}}^H\![k]]
\mathbf{w}_{{\rm RF},\mathcal{T}_{l}}(\!\mathbf{w}_{{\rm RF},\mathcal{T}_{l}}^H\mathbb{E}[\mathbf{y}_{\mathcal{T}_{l}}\![k]
\mathbf{y}_{\mathcal{T}_{l}}^H[k]]
\\\times&\mathbf{w}_{{\rm RF},\mathcal{T}_{l}})^{-1}\mathbf{w}_{{\rm RF},\mathcal{T}_{l}}^H\mathbb{E}[\mathbf{y}_{\mathcal{T}_{l}}[k]
\mathbf{y}_{\mathcal{T}_{l}}^H[k]]\mathbf{W}_{{\rm opt},\mathcal{T}_{l}}[k]).
\end{aligned}
\end{equation}
Note that minimizing (\ref{min_dy}) is equivalent to maximizing the following function
\begin{equation}\label{max_dy}
\begin{aligned}
&\sum_{k=1}^{K}\!\sum_{{l}=1}^{N_r^{\rm RF}}\!\text{Tr}(\mathbf{W}^H_{{\rm opt},\mathcal{T}_{l}}\![k]\mathbb{E}[\mathbf{y}_{\mathcal{T}_{l}}\![k]
\mathbf{y}^H_{\mathcal{T}_{l}}\![k]]\mathbf{w}_{{\rm RF},\mathcal{T}_{l}}
(\!\mathbf{w}^H_{{\rm RF},\mathcal{T}_{l}}
\\&\times\mathbb{E}[\mathbf{y}_{\mathcal{T}_{l}}\![k]
\mathbf{y}^H_{\mathcal{T}_{l}}\![k]]\mathbf{w}_{{\rm RF},\mathcal{T}_{l}}\!)^{-1}
\mathbf{w}^H_{{\rm RF},\mathcal{T}_{l}}
\mathbb{E}[\mathbf{y}_{\mathcal{T}_{l}}\![k]
\mathbf{y}^H_{\mathcal{T}_{l}}\![k]]\mathbf{W}_{{\rm opt},\mathcal{T}_{l}}\![k]\!).
\end{aligned}
\end{equation}
Furthermore, we consider
\begin{equation}\label{w_ass}
\begin{aligned}
\mathbb{E}[\mathbf{y}_{\mathcal{T}_{l}}[1]
\mathbf{y}_{\mathcal{T}_{l}}^H[1]]\approx\cdots\approx\mathbb{E}[\mathbf{y}_{\mathcal{T}_{l}}[K]
\mathbf{y}_{\mathcal{T}_{l}}^H[K]]&\approx\mathbb{E}[\mathbf{y}_{\mathcal{T}_{l}}
\mathbf{y}_{\mathcal{T}_{l}}^H],
\\&1\leq {l}\leq N_r^{\rm RF}.
\end{aligned}
\end{equation}
Note that the approximation error in (\ref{w_ass}) can be ignored in large antennas regime as proven in Appendix D. By substituting (\ref{w_ass}) into (\ref{max_dy}), we can obtain
\begin{equation}\label{w_lambda}
\begin{aligned}
&\sum_{k=1}^{K}\!\!\sum_{{l}=1}^{N_r^{\rm RF}}\!\!\text{Tr}(\mathbf{W}_{{\rm opt},\mathcal{T}_{l}}^H[k]\mathbb{E}[\mathbf{y}_{\mathcal{T}_{l}}
\mathbf{y}_{\mathcal{T}_{l}}^H]\mathbf{w}_{{\rm RF},\mathcal{T}_{l}}
(\mathbf{w}_{{\rm RF},\mathcal{T}_{l}}^H\mathbb{E}[\mathbf{y}_{\mathcal{T}_{l}}
\mathbf{y}_{\mathcal{T}_r}^H]
\\&\times\mathbf{w}_{{\rm RF},\mathcal{T}_{l}})^{-1}\mathbf{w}_{{\rm RF},\mathcal{T}_{l}}^H
\mathbb{E}[\mathbf{y}_{\mathcal{T}_{l}}
\mathbf{y}_{\mathcal{T}_{l}}^H]\mathbf{W}_{{\rm opt},\mathcal{T}_{l}}[k])
\\=&\sum_{k=1}^{K}\!\!\sum_{{l}=1}^{N_r^{\rm RF}}\frac{||\mathbf{W}_{{\rm opt},\mathcal{T}_{l}}^H\![k]\mathbb{E}[\mathbf{y}_{\mathcal{T}_{l}}\!
	\mathbf{y}_{\mathcal{T}_{l}}^H]\mathbf{w}_{{\rm RF},\mathcal{T}_{l}}||_2^2}{||\mathbb{E}[\mathbf{y}_{\mathcal{T}_{l}}
	\mathbf{y}_{\mathcal{T}_{l}}^H]^{\frac{1}{2}}\mathbf{w}_{{\rm RF},\mathcal{T}_{l}}||_2^2}
\\=&\sum_{{l}=1}^{N_r^{\rm RF}}\!\!\frac{\mathbf{w}_{{\rm RF},\mathcal{T}_{l}}^H\mathbb{E}[\mathbf{y}_{\mathcal{T}_{l}}
	\mathbf{y}_{\mathcal{T}_{l}}^H]^{\frac{1}{2}}\mathbf{W}_{{\rm opt},\mathcal{T}_{l}}\![k]\mathbf{W}_{{\rm opt},\mathcal{T}_{l}}^H\![k]\mathbb{E}[\mathbf{y}_{\mathcal{T}_{l}}
	\mathbf{y}_{\mathcal{T}_{l}}^H]^{\frac{1}{2}}\mathbf{w}_{{\rm RF},\mathcal{T}_{l}}}{||\mathbb{E}[\mathbf{y}_{\mathcal{T}_r}
	\mathbf{y}_{\mathcal{T}_{l}}^H]^{\frac{1}{2}}\mathbf{w}_{{\rm RF},\mathcal{T}_{l}}||_2^2}.
\end{aligned}
\end{equation}
The maximum of (\ref{w_lambda}) is
$\sum_{{l}=1}^{N_r^{\rm RF}}\lambda_1^2(\mathbf{W}_{\mathcal{T}_{l}})$ when $\mathbf{w}_{{\rm RF},\mathcal{T}_l}=\frac{1}{{\sqrt {{N_r^{\rm sub}}} }}{\rm exp}(j\angle(\mathbf{u}_{\mathcal{T}_l,1}))$ as discussed in Section \ref{HDCD}. Therefore, the optimization of $\{\mathcal{T}_{l}\}_{{l}=1}^{N_r^{\rm RF}}$ can be formulated as\footnote{Note that (\ref{w_pro_dy}) is more difficult to reach than (\ref{pro_dy}) because of the existence of the weight matrix $\mathbb{E}[\mathbf{y}_{\mathcal{T}_{l}}[k]
\mathbf{y}_{\mathcal{T}_{l}}^H[k]]^{\frac{1}{2}}$.}
\begin{equation}\label{w_pro_dy}
\setlength{\abovedisplayskip}{10pt}
\setlength{\belowdisplayskip}{10pt}
\begin{aligned}
\max\limits_{\mathcal{T}_1,\cdots,\mathcal{T}_{N_r^{\rm RF}}}&\sum\nolimits_{{l}=1}^{N_r^{\rm RF}}\lambda_1^2(\mathbf{W}_{\mathcal{T}_{l}})
\\&\text{s.t. }\cup_{{l}=1}^{N_r^{\rm RF}}\mathcal{T}_{l}=\{1,\cdots,N_r\},
\\&\mathcal{T}_i\cap\mathcal{T}_j=\emptyset\text{ for }i\not=j,\ \mathcal{T}_{l}\not=\emptyset\ \forall {l}.
\end{aligned}
\end{equation}

Similar to (\ref{lambda_approx}), we further obtain $\lambda_1^2(\mathbf{W}_{\mathcal{T}_{l}})=\lambda_1(\mathbf{R}_{\mathcal{T}_{l}})$ and
\begin{equation}
\setlength{\abovedisplayskip}{10pt}
\setlength{\belowdisplayskip}{10pt}
\begin{aligned}
\lambda_1(\mathbf{R}_{\mathcal{T}_{l}})&\approx \frac{1}{\text{card}(\mathcal{T}_{l})}\sum\nolimits_{i=1}^{\text{card}(\mathcal{T}_{l})}\sum\nolimits_{j=1}^{\text{card}(\mathcal{T}_{l})}|[\mathbf{R}_{\mathcal{T}_{l}}]_{i,j}|
\\&=\frac{1}{\text{card}(\mathcal{T}_{l})}\sum_{i\in\mathcal{T}_{l} }\sum_{j\in\mathcal{T}_{l}}|[\mathbf{R}_W]_{i,j}|,
\end{aligned}
\end{equation}
where $\mathbf{R}_{\mathcal{T}_{l}}=\mathbf{W}_{\mathcal{T}_{l}}\mathbf{W}_{\mathcal{T}_{l}}^H$ and $\mathbf{R}_W=\mathbf{W}\mathbf{W}^H$. Finally, the antenna grouping $\{\mathcal{T}_{l}\}_{{l}=1}^{N_r^{\rm RF}}$ at the RX can be obtained using Algorithm \ref{alg:AHC} by replacing the input
parameters $\mathbf{R}_F$, $N_t$, and $N_t^{\rm RF}$ for the TX
with $\mathbf{R}_W$, $N_r$, and $N_r^{\rm RF}$ for the RX.
By contrast, the antenna grouping solution in \cite{170825} considers the hybrid transmit array but the fully-digital receive array, and how to group
the antennas for hybrid receive array is not explicitly specified.


\vspace*{-4mm}
\section{Performance Evaluation}
\begin{figure}[t]
    	\centering
	\includegraphics[width=1\columnwidth, keepaspectratio]{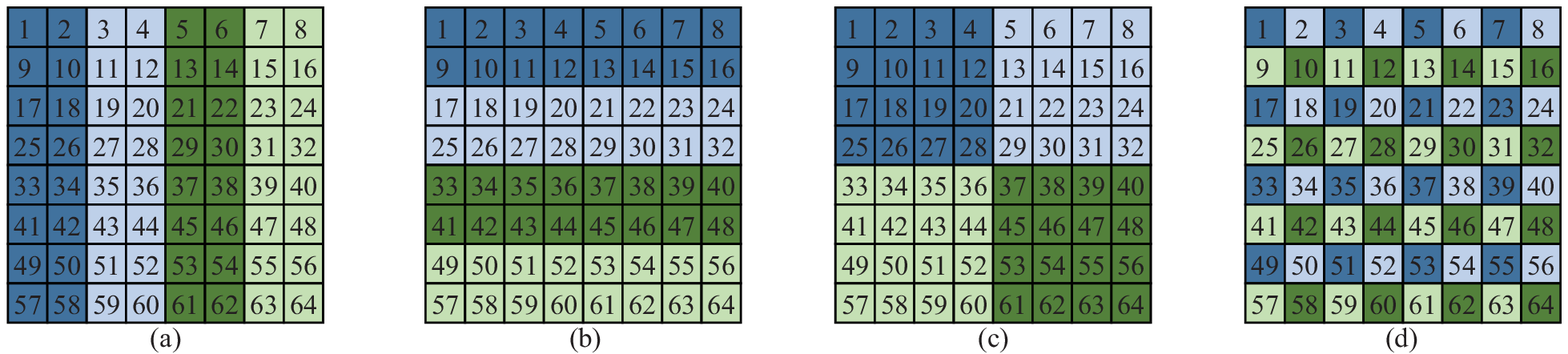}
\vspace*{-12mm}
	\caption{Four types of typical FS patterns: (a) Vertical type; (b) Horizontal type; (c) Squared type; (d) Interlaced type.}\label{FS}
\vspace*{-8mm}
\end{figure}
In this section, we will investigate the SE, EE, and BER performance of the proposed hybrid precoder/combiner design.
In simulations, the pulse shaping filter is $p(\tau ){\rm{ = }}\delta \left( \tau  \right)$, the length of cyclic prefix is $D=64$,
and the number of subcarriers is $K=512$.
The carrier frequency is $f_c=30$ GHz, the bandwidth is $B_s=500$ MHz, the path delay is uniformly distributed in $[0,DT_s]$ ($T_s=1/B_s$ is the symbol period),
the number of clusters is $N_{\rm cl}=8$, the azimuth/elevation angle spread of each cluster is $7.5^{\circ}$ for both
AoD and AoA, and there are $N_{\rm ray}=10$ rays within each cluster.
Both the TX and RX adopt $8\times8$ hybrid UPA with $N_t^{\rm RF}=N_r^{\rm RF}=4$ unless otherwise stated. The number of data stream is $N_s=3$.
Four types of classical FS patterns  shown in Fig. \ref{FS} will be investigated, where the antennas with the same color are connected to the same RF chain for constituting a subarray. The channel estimation overhead is not considered in the evaluation.

State-of-the-art solutions will be compared as benchmarks. 1) \textbf{Optimal fully-digital} scheme considers the fully-digital MIMO system, where the SVD-based precoder/combiner is adopted as the performance upper bound. 2) Since the OMP-based spatially sparse precoding \cite{OMP} is proposed for narrowband channels, we consider a broadband version that can simultaneously design the RF precoder/combiner for all subcarriers, denoted by \textbf{simultaneous OMP (SOMP)} scheme. 3) \textbf{Discrete Fourier transform (DFT) codebook} scheme~\cite{170825} designs the RF precoder/combiner from the DFT codebook instead of steering vectors codebook in SOMP scheme~\cite{DFT_cb}. 4) \textbf{Covariance eigenvalue decomposition (EVD)} scheme
focuses on the hybrid precoder design based on the EVD of the channel covariance matrix $\mathbf{R}_{\rm cov}=\frac{1}{K}\sum_{k=1}^{K}\mathbf{H}^H[k]\mathbf{H}[k]$ and assumes the fully-digital receive array.
To extend this scheme to hybrid receive arrays, the RF combiner is designed using the same processing as the RF precoder by replacing $\mathbf{R}_{\rm cov}$ with $\mathbf{\tilde R}_{\rm cov}=\frac{1}{K}\sum_{k=1}^{K}\mathbf{H}[k]\mathbf{H}^H[k]$, and the digital combiner is designed based on MMSE criterion.

\subsection{SE and BER Performance Evaluation}
\begin{figure}[tb]
    \centering \subfigure{\includegraphics[width=\columnwidth, keepaspectratio]{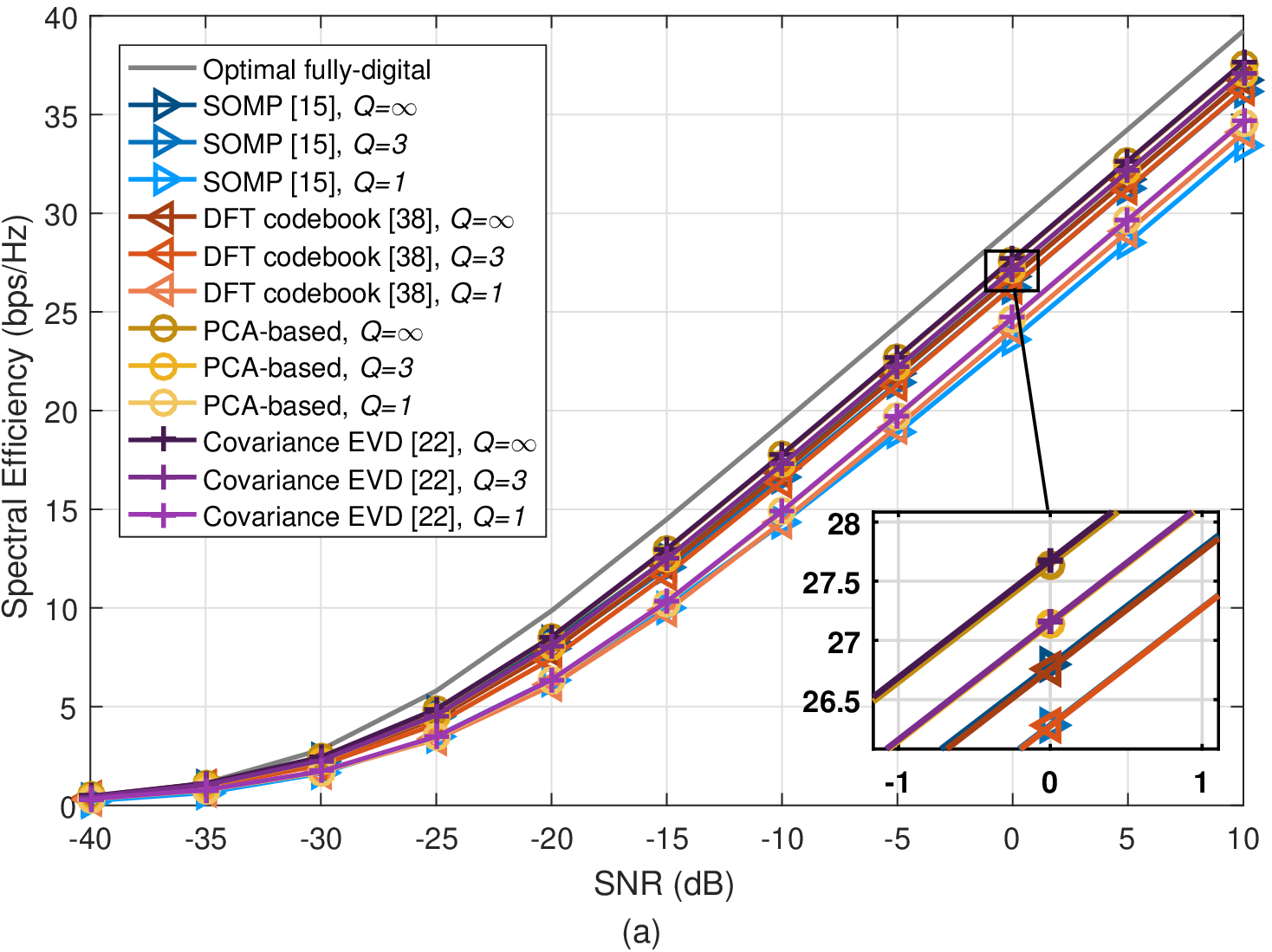}}\\
    \subfigure{\includegraphics[width=\columnwidth, keepaspectratio]{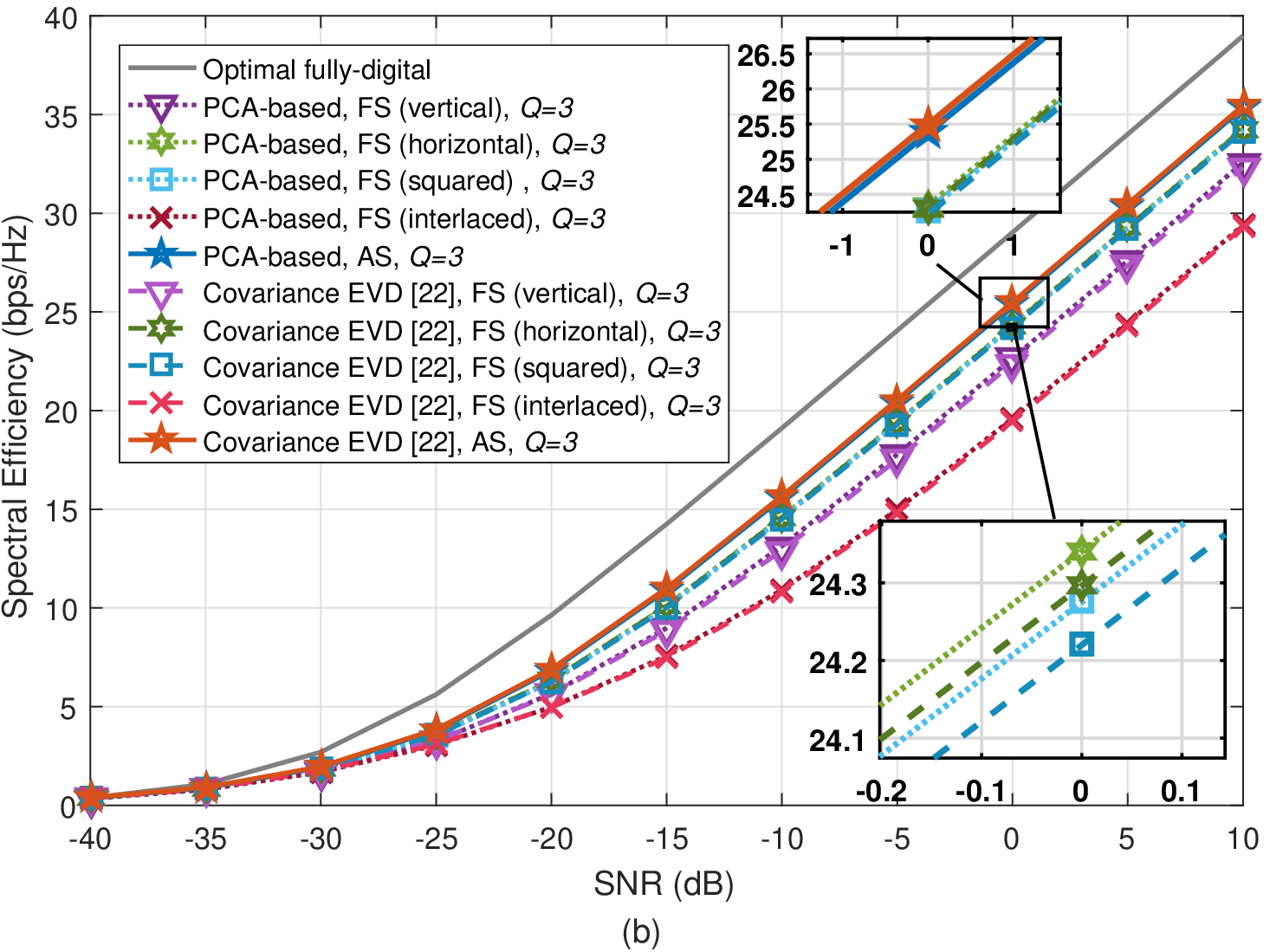}}\\
	\caption{SE performance comparison of different hybrid precoder schemes, where the hybrid transmit array and fully-digital receive array are considered: (a) FCA with $Q=\infty$, $Q=3$, and $Q=1$; (b) PCS with $Q=3$.}\label{fig_pre}
\end{figure}
In Fig. \ref{fig_pre}, we compare the SE performance of different hybrid precoder schemes,
where the system adopts the hybrid transmit array and fully-digital receive array.
For the hybrid transmit array, FCA and PCS are investigated in Fig. \ref{fig_pre} (a) and (b), respectively.
Besides, the practical phase shifters with resolutions $Q=3$ and $Q=1$ are considered, and the ideal
phase shifters without quantization, denoted by $Q=\infty$, is also compared to examine the impact of $Q$ \cite{quan}.
For the FCA, Fig. \ref{fig_pre} (a) shows that the proposed PCA-based solution and the covariance EVD-based solution have the similar performance,
and our proposed solution has the considerable superiority over conventional DFT codebook-based and SOMP-based solutions.
This is because that our proposed solution exploits the principal components of the common column
space of $\left\{ {{\bf{H}}[k]} \right\}_{k = 1}^K$ to design the analog precoder $\mathbf{F}_{\rm RF}$.
By contrast, the analog precoders designed by the SOMP-based and DFT codebook-based solutions are
based on the codebooks, whose entries are limited to the steering vector form. This kind of inflexible RF precoder designs
would lead to the poor performance.
Additionally, for the proposed hybrid precoder scheme,
we can observe the performance gap between adopting $Q=\infty$ and $Q=3$ is negligible,
 but that between adopting $Q=3$ and $Q=1$ is around 3 (bps/Hz) at high SNR conditions.
As for the PCS, we only investigate the SE performance under $Q=3$ in Fig. \ref{fig_pre} (b),
which manifests that our proposed scheme is as good as the state-of-the-art covariance EVD-based hybrid precoder design for different FS patterns and AS.
Finally,
the performance gain achieved by the proposed antenna grouping algorithm for AS over several typical FS patterns is more than 1 (bps/Hz) at high SNR conditions.
So the advantage by using AS is self-evident.
\begin{figure}[tb]
    \centering \subfigure{\includegraphics[width=\columnwidth, keepaspectratio]{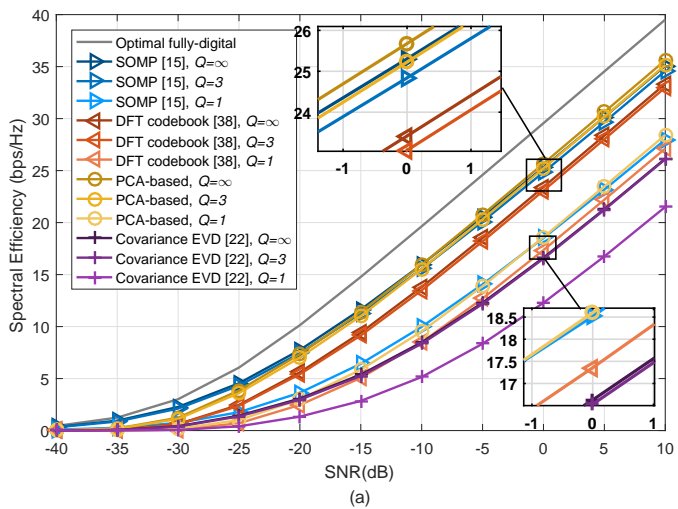}}\\
    \subfigure{\includegraphics[width=\columnwidth, keepaspectratio]{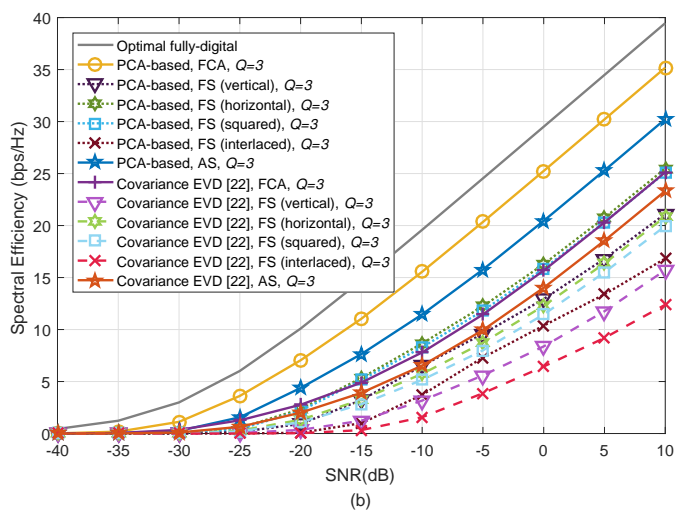}}
	\caption{SE performance comparison of different hybrid precoder/combiner schemes, where both TX and RX employ the hybrid array: (a) FCA with $Q=\infty$, $Q=3$, and $Q=1$; (b) PCS with $Q=3$.}\label{fig_SE}
\end{figure}

Fig. \ref{fig_SE} examines the SE performance of different hybrid precoder/combiner schemes,
where both the TX and RX adopt the hybrid arrays, and the phase shifters under different $Q$ are considered.
From Fig. \ref{fig_SE} (a),
we can observe that the proposed solution is superior to
other three state-of-art solutions,
and the performance of covariance EVD-based scheme becomes poor. When expanding to RX with hybrid combiner, the covariance EVD-based scheme  performs poorly because \cite{170825} initially considers the fully-digital combiner and its extension to hybrid combiner by using the channel reciprocity suffers from a large performance loss. While in our proposed scheme, both the hybrid precoder and hybrid combiner are jointly designed based on PCA framework so the better SE performance can be achieved.
Especially,
for the proposed solution, the performance gap between adopting $Q=\infty$ and $Q=3$ is negligible,
 but that between adopting $Q=3$ and $Q=1$ is larger than 5 (bps/Hz) at high SNR conditions.
%
For the PCS under $Q=3$, Fig. \ref{fig_SE} (b) shows the superiority of
our scheme over state-of-the-art covariance EVD-based scheme for different FS
patterns and AS. This is because the antenna grouping scheme in \cite{170825} is based on the greedy search so that the local optimal solution may be acquired.
This could lead to the extremely unbalanced antenna grouping case that no antenna is assigned to one RF chain, and thus the SE performance is degraded.
By contrast, the proposed antenna grouping algorithm introduces the mutual correlation metric (\ref{mul_cor}),
which can effectively avoid this issue. Besides, Fig. 4 (b) also indicates that at least 5 (bps/Hz) SE gains can be achieved by the AS over the FS,
since the proposed shared-AHC algorithm can group the antennas adapted to the CSI for the enhanced SE performance. 

In Fig. \ref{fig_BER}, we compare the BER performance of different hybrid precoder/combiner schemes, where the same channel parameters as considered in Fig. \ref{fig_SE} are used, and 16 QAM is adopted for transmission.
From Fig. \ref{fig_BER}, similar
conclusions to those observed for Fig. \ref{fig_SE} can be obtained.
In particular, it can be seen that the AS by using our proposed antenna grouping scheme significantly outperforms four typical FS patterns.
\begin{figure}[tb]
    \centering \subfigure{\includegraphics[width=\columnwidth, keepaspectratio]{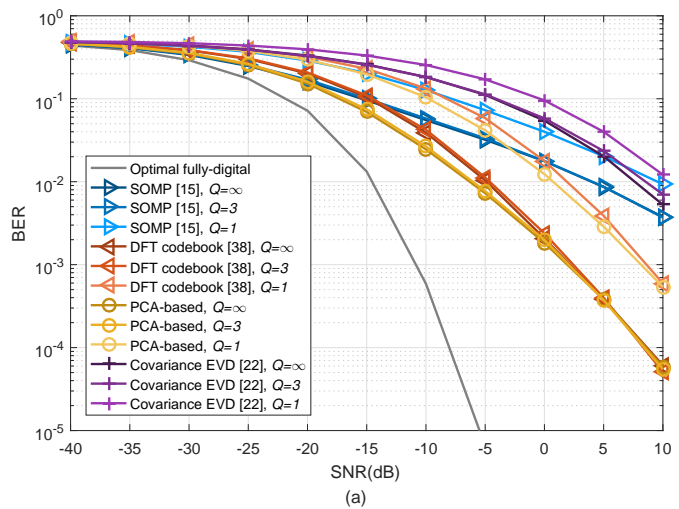}}\\
        \subfigure{\includegraphics[width=\columnwidth, keepaspectratio]{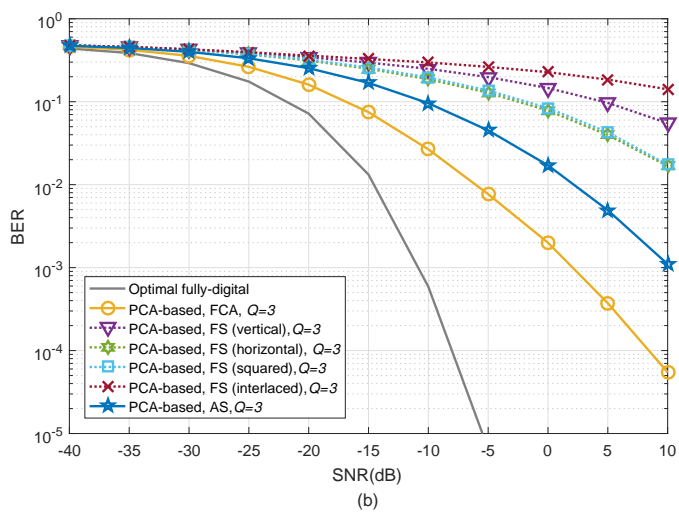}}
	\caption{BER performance comparison of different hybrid precoder/combiner schemes,  where both TX and RX employ
the hybrid array: (a) FCA with $Q=\infty$, $Q=3$ and $Q=1$; (b) PCS with typical quantization $Q=3$.}\label{fig_BER}
\end{figure}

\subsection{EE Performance Evaluation}
{To further explore the performance of different array structure, we analyze the EE performance of different array structure in this subsection.} The EE metric is defined as $\eta =RB_s/P$, where $B_s$ is the transmission bandwidth, $R$ is the SE in (\ref{SE}), and $P$ is the total power consumption of the antenna arrays.
Here $P$ depends on the following two factors.
Firstly, FCA and PCS have the different power consumption due to the different numbers of phase shifters required, where
the FCA requires $N_{\rm PS}=N_tN_t^{\rm RF}=64\times 4$ phase shifters, but the PCS only requires $N_{\rm PS}=N_t=64$ phase shifters.
Secondly, passive and active arrays have different power consumption, since they have different architectures. From Fig. \ref{fig_ant_str} (a) and (b), 
we can observe that both of them consist of analog-digital/digital-analog convertors (AD/DA), low noise amplifiers (LNA), power amplifiers (PA), local oscillators (LO), duplexers or switches (DPX/S)\footnote{DPX/S is used to transmit/receive signals
by sharing the same antenna hardware. When working in transmitting (receiving) mode, the DPX/S ensures
the transmit (receive) signal delivered from PA (antennas) to antennas (LNA). 
So the power consumption of DPX/S
can be ignored since the switching duration between the transmitting mode and receiving mode is neglected. Besides, the power consumption of switches when the antenna grouping patterns for AS changes is also ignored.},
and mixers etc.
However, they have the different numbers of PAs/LNAs. For passive antennas, the number of PAs/LNAs is the same as that of RF chains. While for active antennas, the number of PAs/LNAs is the same as that of antennas. This difference can lead to the different power consumption. 
The power values of electronic components that dominate the power consumption are listed as follows:
$P_{\rm PS}=15$ mW for 3-bit phase shifter \cite{27}, $P_{\rm AD}=P_{\rm DA}=200$ mW for AD/DA \cite{27}, $P_{\rm mix}=39$ mW for
mixer \cite{39}, $P_{\rm PA}=138$ mW for PA \cite{36},
$P_{\rm LNA}=39$ mW for LNA \cite{36}, $P_{\rm LO}=5$ mW for LO \cite{27}, and $P_{\rm syn}=50$ mW for synchronizer \cite{syn}.
\begin{figure*}[t]
	\centering
	\includegraphics[width=1.8\columnwidth, keepaspectratio]{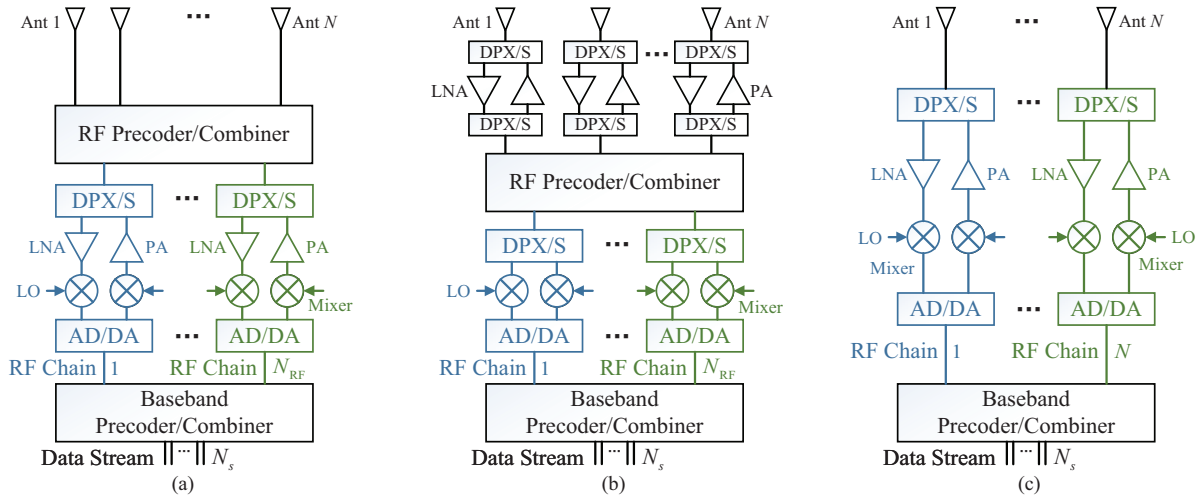}
	\caption{
(a) Hybrid MIMO with passive antennas; (b) Hybrid MIMO with active antennas; (c) Fully-digital MIMO array~\cite{ant_str}.}
\label{fig_ant_str}
\end{figure*}
Therefore, the power consumption for FCA and PCS adopting active antenna architecture can be respectively calculated as
\begin{equation}\nonumber
\begin{aligned}
P_{\rm FCA}^p\!\!&=\!N_t^{\rm RF}\!(P_{\rm DA}\!+\!P_{\rm mix}\!+\!P_{\rm LO}\!+\!P_{\rm PA})\!+\!N_tN_t^{\rm RF}P_{\rm PS}+\!2P_{\rm syn}
\\&\!+\!N_r^{\rm RF}\!(P_{\rm AD}\!+\!P_{\rm mix}\!+\!P_{\rm LO}\!+\!P_{\rm LNA})
\!+\!N_rN_r^{\rm RF}P_{\rm PS}\!,
\end{aligned}
\end{equation}
\begin{equation}\nonumber
\begin{aligned}
P_{\rm PCS}^p\!&=\!N_t^{\rm RF}(P_{\rm DA}\!+\!P_{\rm mix}\!+\!P_{\rm LO}\!+\!P_{\rm PA})\!+\!N_tP_{\rm PS}+\!2P_{\rm syn}
\\&\!+\!N_r^{\rm RF}(P_{\rm AD}\!+\!P_{\rm mix}\!+\!P_{\rm LO}\!+\!P_{\rm LNA})
\!+\!N_rP_{\rm PS}\!.
\end{aligned}
\end{equation}
Moreover, the power consumption for FCA and PCS adopting active antenna architecture are
\begin{equation}\nonumber
\begin{aligned}
P_{\rm FCA}^a\!\!&=\!N_t^{\rm RF}\!(P_{\rm DA}\!\!+\!P_{\rm mix}\!\!+\!P_{\rm LO}\!)\!+\!N_t\!N_t^{\rm RF}\!P_{\rm PS}\!+\!N_t\!P_{\rm PA}\!\!
+\!2P_{\rm syn}
\\&\!+\!N_r^{\rm RF}\!(P_{\rm AD}\!\!+\!P_{\rm mix}\!\!+\!P_{\rm LO}\!)\!+\!N_r\!N_r^{\rm RF}\!P_{\rm PS}\!+\!N_r\!P_{\rm LNA},
\end{aligned}
\end{equation}
\begin{equation}\nonumber
\begin{aligned}
P_{\rm PCS}^a\!&=\!N_t^{\rm RF}(P_{\rm DA}\!+\!P_{\rm mix}\!+\!P_{\rm LO})\!+\!N_tP_{\rm PS}\!\!+\!N_tP_{\rm PA}
\!+\!2P_{\rm syn}
\\&\!+\!N_r^{\rm RF}(P_{\rm AD}\!+\!P_{\rm mix}\!+\!P_{\rm LO})\!+\!N_rP_{\rm PS}
\!+\!N_rP_{\rm LNA}.
\end{aligned}
\end{equation}
Besides, for fully-digital array (FDA) as shown in Fig. 2 (c), the power consumption is
\begin{equation}\nonumber
\begin{aligned}
P_{\rm FDA}&=N_t(P_{\rm PA}+P_{\rm DA}+P_{\rm mix}+P_{\rm LO})+2P_{\rm syn}
\\&+N_r(P_{\rm LNA}+P_{\rm AD}+P_{\rm mix}+P_{\rm LO}).
\end{aligned}
\end{equation}

\begin{figure}[tb]
    \centering \subfigure{\includegraphics[width=\columnwidth, keepaspectratio]{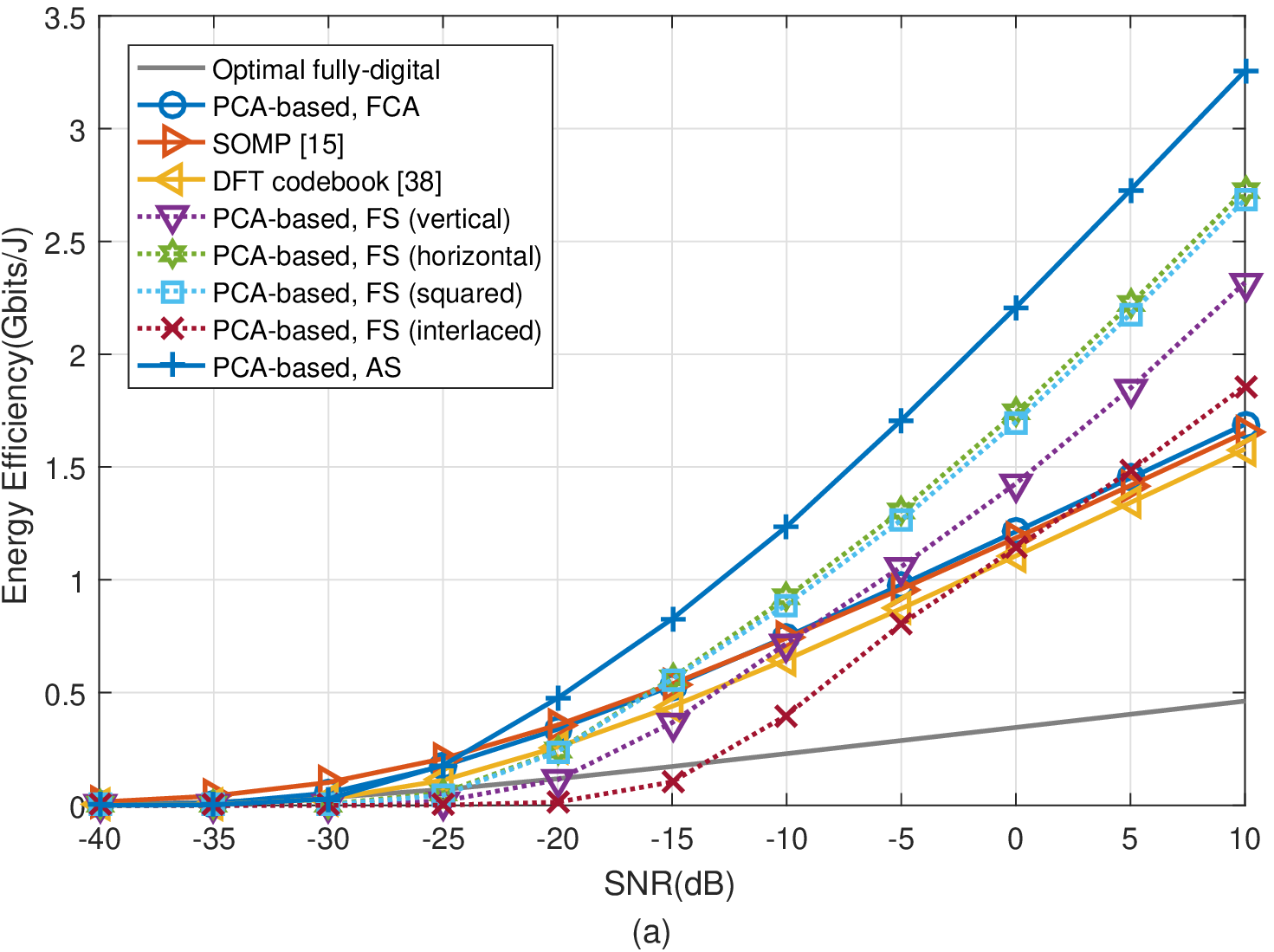}} \\ 
    \subfigure{\includegraphics[width=\columnwidth, keepaspectratio]{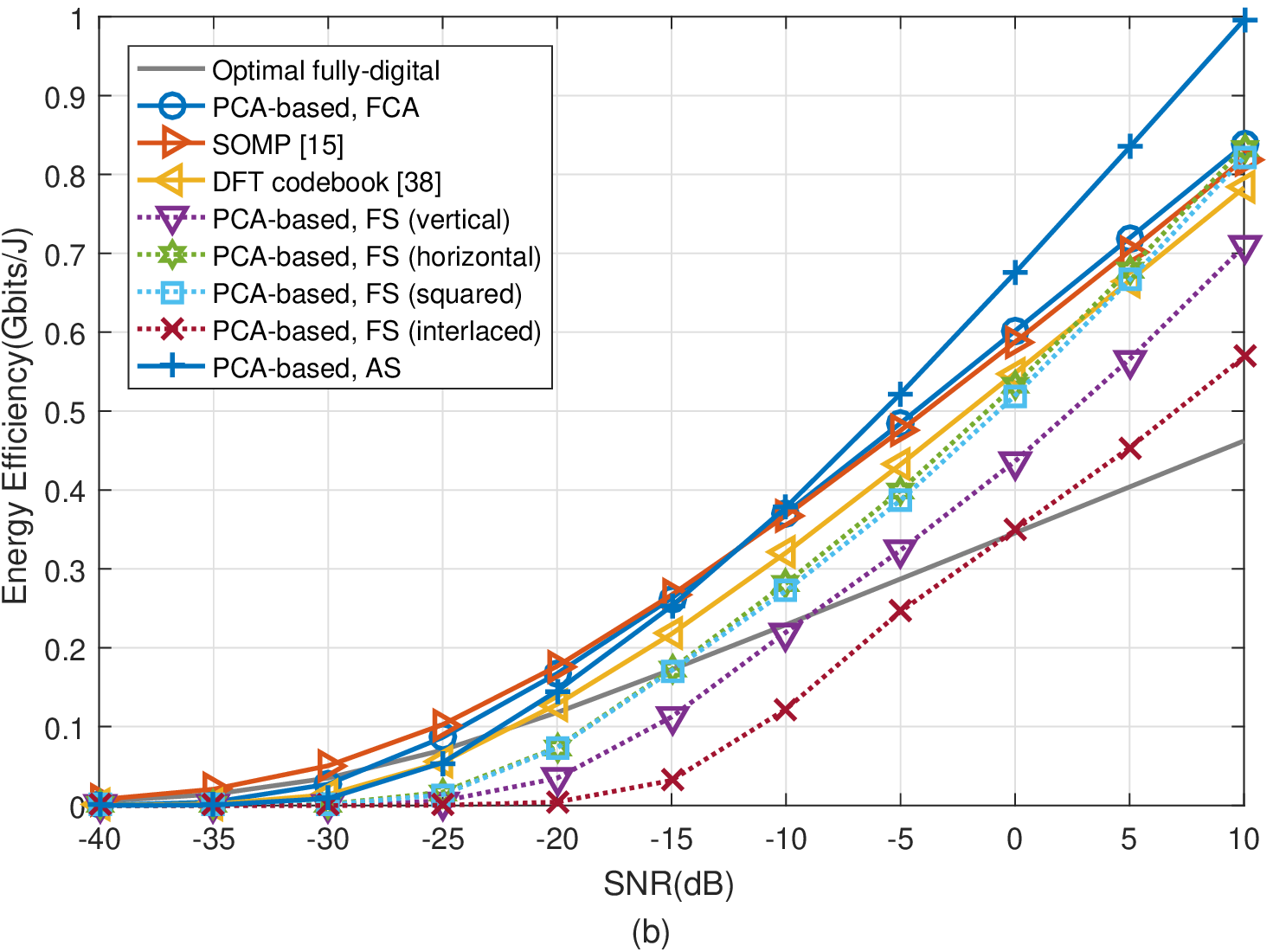}} 
	\caption{EE performance comparison of different hybrid precoder/combiner schemes based on different antenna architectures with $Q=3$: (a) Passive antenna; (b) Active antenna.}\label{fig_EE}
\end{figure}

In Fig. \ref{fig_EE}, we compare the EE performance of different hybrid precoder/combiner schemes under $Q=3$,
where passive and active antenna architectures are investigated in Fig. \ref{fig_EE} (a) and (b), respectively. 
For passive antenna architecture examined in Fig. \ref{fig_EE} (a),
the EE performance of PCS by using the proposed PCA-based
hybrid precoder/combiner scheme
outperforms that of FCA by using the proposed and other state-of-the-art schemes.
The reason is that PCS adopts a much smaller number of phase shifters than FCA, though the SE performance
achieved by PCS is inferior to that achieved by FCA.
Moreover, AS obviously outperforms the other FS patterns in SE, and it consumes the very similar power with the other FS patterns.
Therefore, AS
outperforms other four types of FS patterns in EE.
Note that the optimal fully-digital scheme has the worst EE performance,
since the numbers of power-consuming PAs, LNAs, ADs/DAs, mixers are
proportional to that of antennas.

For active antennas investigated in Fig. \ref{fig_EE} (b),
the EE advantage for different FS patterns by using the proposed hybrid precoder/combiner
scheme over the FCA with several typical hybrid precoder/combiner schemes and fully-digital array with the optimal precoder/combiner scheme is not obvious.
This is because the active antenna architecture
requires the power-hungry PAs/LNAs for each antenna.
Meanwhile, the advantage of the reduced power consumption benefiting from FS structure is greatly weakened by its disadvantage in SE performance
when compared with FCA.
Finally, the EE performance of AS with the proposed hybrid precoder/combiner scheme still has the considerable advantage
over the other schemes. This observation reveals the appealing
advantage of AS in practical situation when both the power consumption and SE should be well balanced.

\subsection{Computational Complexity and Robustness of The Proposed Shared-AHC Algorithm}

\begin{figure}[tb]
    \centering \subfigure{\includegraphics[width=\columnwidth, keepaspectratio]{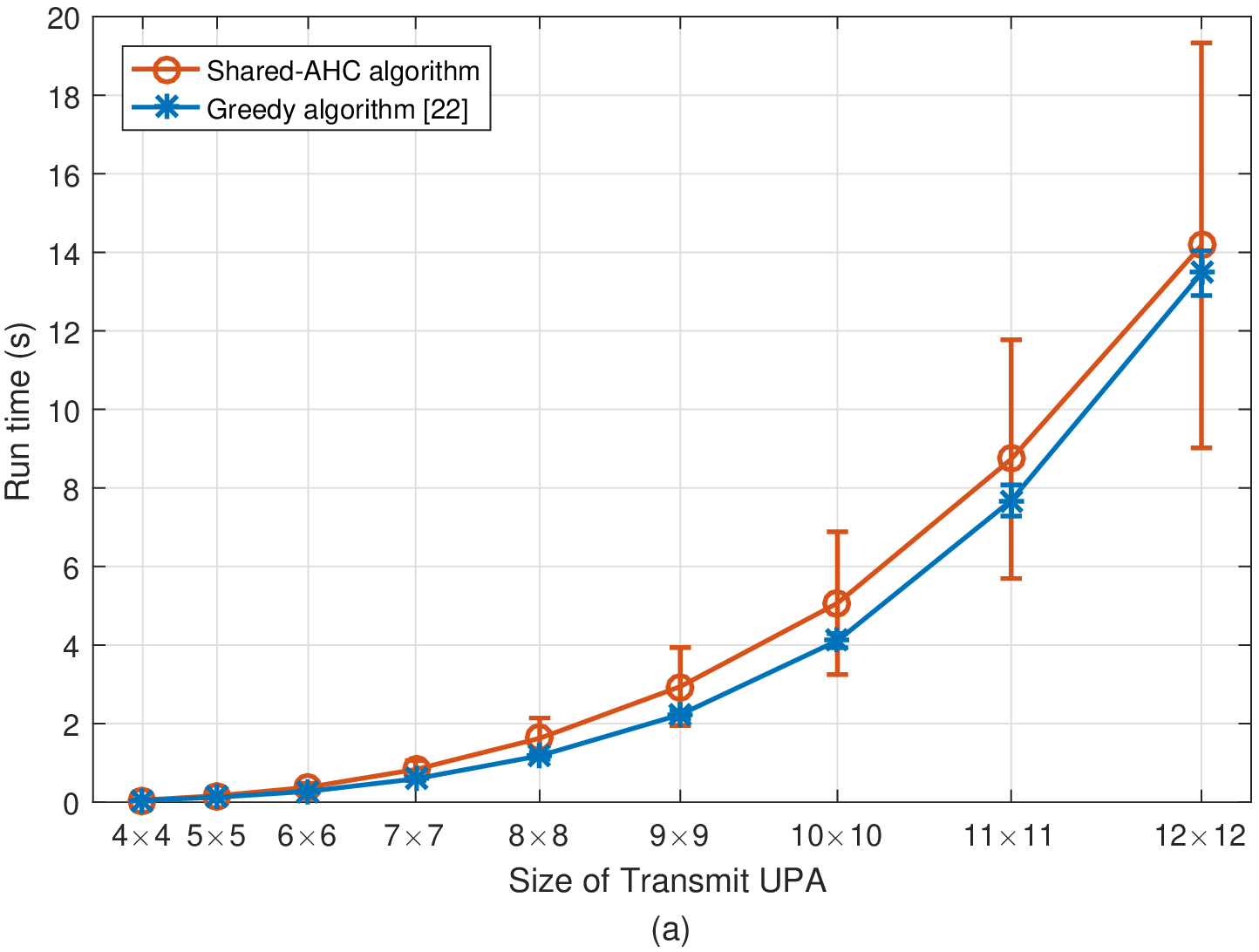}}\\ \subfigure{\includegraphics[width=\columnwidth, keepaspectratio]{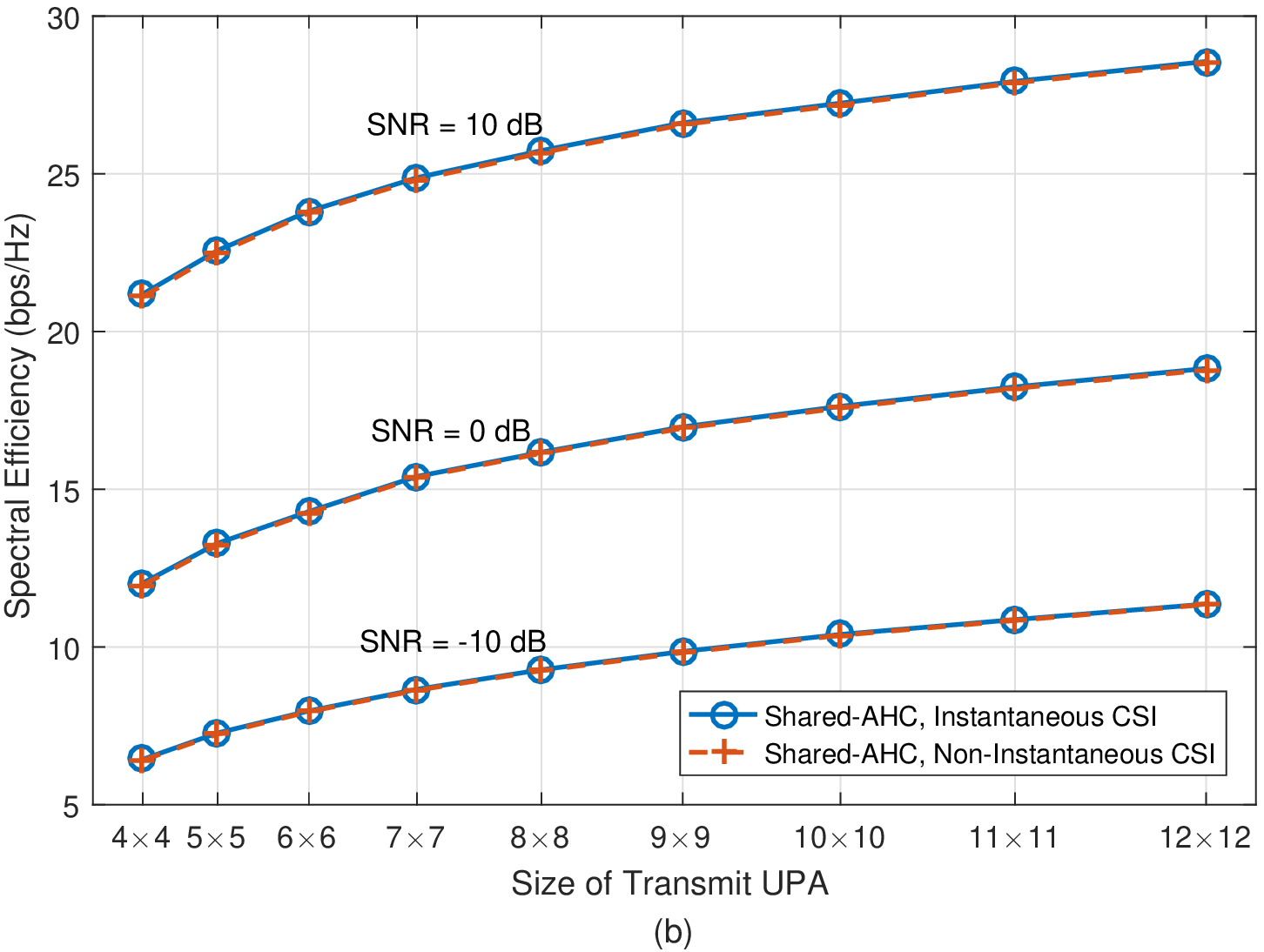}}
    	\caption{(a) Runtime comparison; (b) Robustness of the proposed antenna grouping scheme to time-varying channels.
    Here the same simulation configuration as considered in Fig. \ref{fig_pre} is used except for the size of transmit UPA.}\label{fig_DS}
\end{figure}
With respect of the antenna grouping, the computational complexity of the shared-AHC algorithm for the proposed PCA-based method and the greedy algorithm for the covariance EVD method \cite{170825} are compared.
However, their computational complexity is difficult to be accurately calculated.
On the one hand, both two algorithms have the selection statements of ``if" and ``else", and the complexity can be different for different selections.
On the other hand, the total number of iterations for the shared-AHC algorithm is adaptive.
Therefore, Fig. \ref{fig_DS} (a) compares the practical runtime of these two algorithms instead,
where the simulations are based on the software MATLAB 2016a and the hardware Intel Core i7-7700 CPU and 16 GB RAM.
Fig. \ref{fig_DS} (a) shows the mean and standard deviation of the runtime of two antenna grouping algorithms, and their similar run time versus different sizes of transmit UPA can be observed.

Fig. \ref{fig_DS} (b) investigates the robustness of the shared-AHC algorithm to time-varying channels.
In simulations, we consider the time-varying block-fading channels, where each time block consists of 10 OFDM symbols,
the channels of different OFDM symbols from the same time blocks
are correlated, but the channels from different time blocks are mutually independent.
For time-varying channels, the variation rate of
channel AoAs/AoDs is much slower than that of the channel gains \cite{tvchannel}. So we consider that the
channels of different OFDM symbols in the same time block
share the same AoAs/AoDs and the modulus values of channel gains, but have the mutually independent phase values of channel gains.
In Fig. \ref{fig_DS} (b), the curves labeled with `Instantaneous CSI' indicate the antenna grouping is updated in each OFDM symbol
by using the proposed algorithm based on the instantaneous CSI,
and the curves labeled with `Non-Instantaneous CSI' indicate the antenna grouping is updated in every time block
by using the proposed algorithm based on the CSI of the first OFDM symbol.
We can observe negligible SE performance loss between two groups of curves. Therefore, the robustness of the proposed
antenna grouping algorithm to time-varying channels is confirmed. 

\subsection{Robustness of The Proposed Hybrid Precoder/Combiner Design to Channel Perturbation}
\begin{figure}[tb]
    \centering \subfigure{\includegraphics[width=\columnwidth, keepaspectratio]{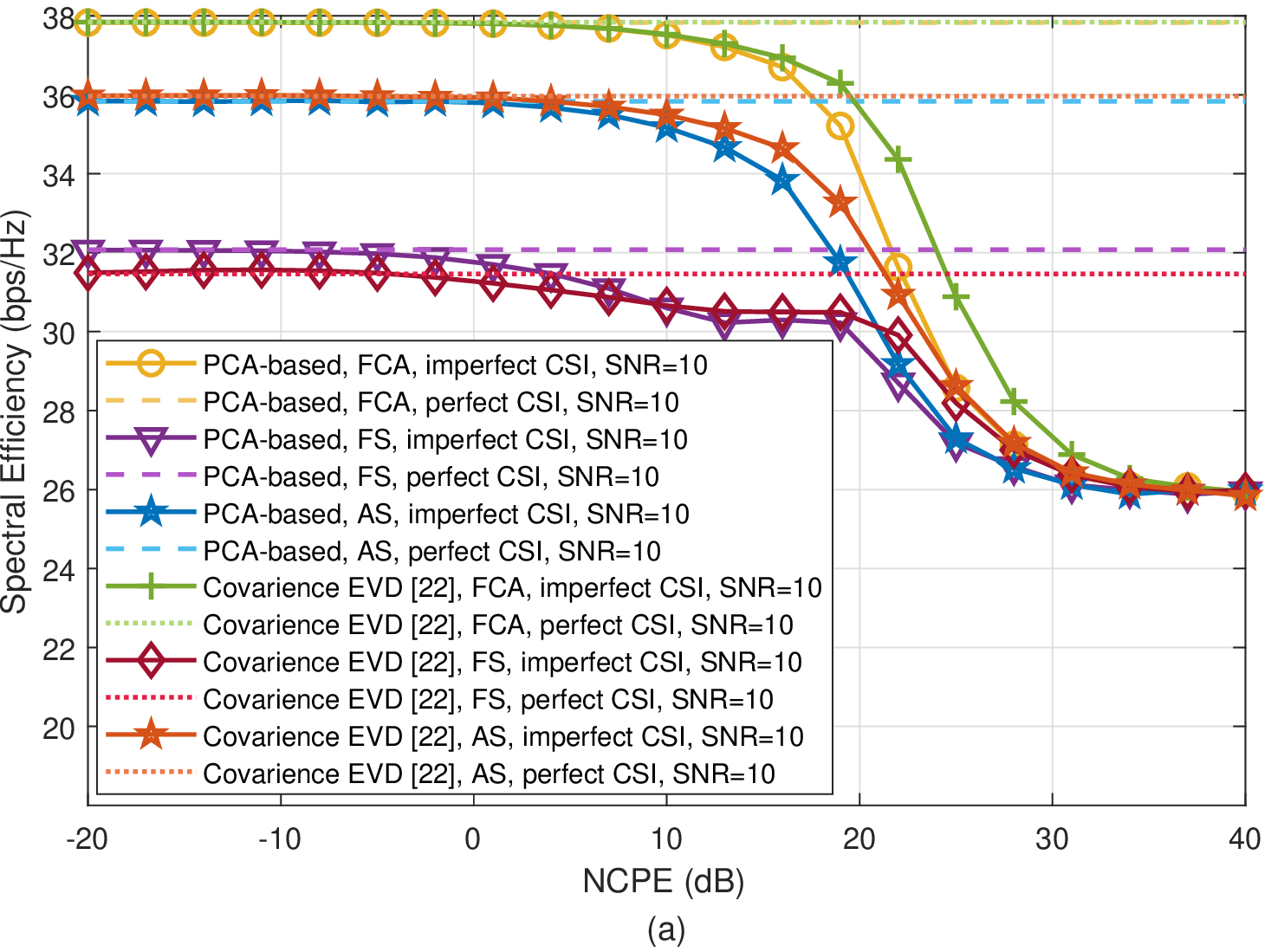}}\\ \subfigure{\includegraphics[width=\columnwidth, keepaspectratio]{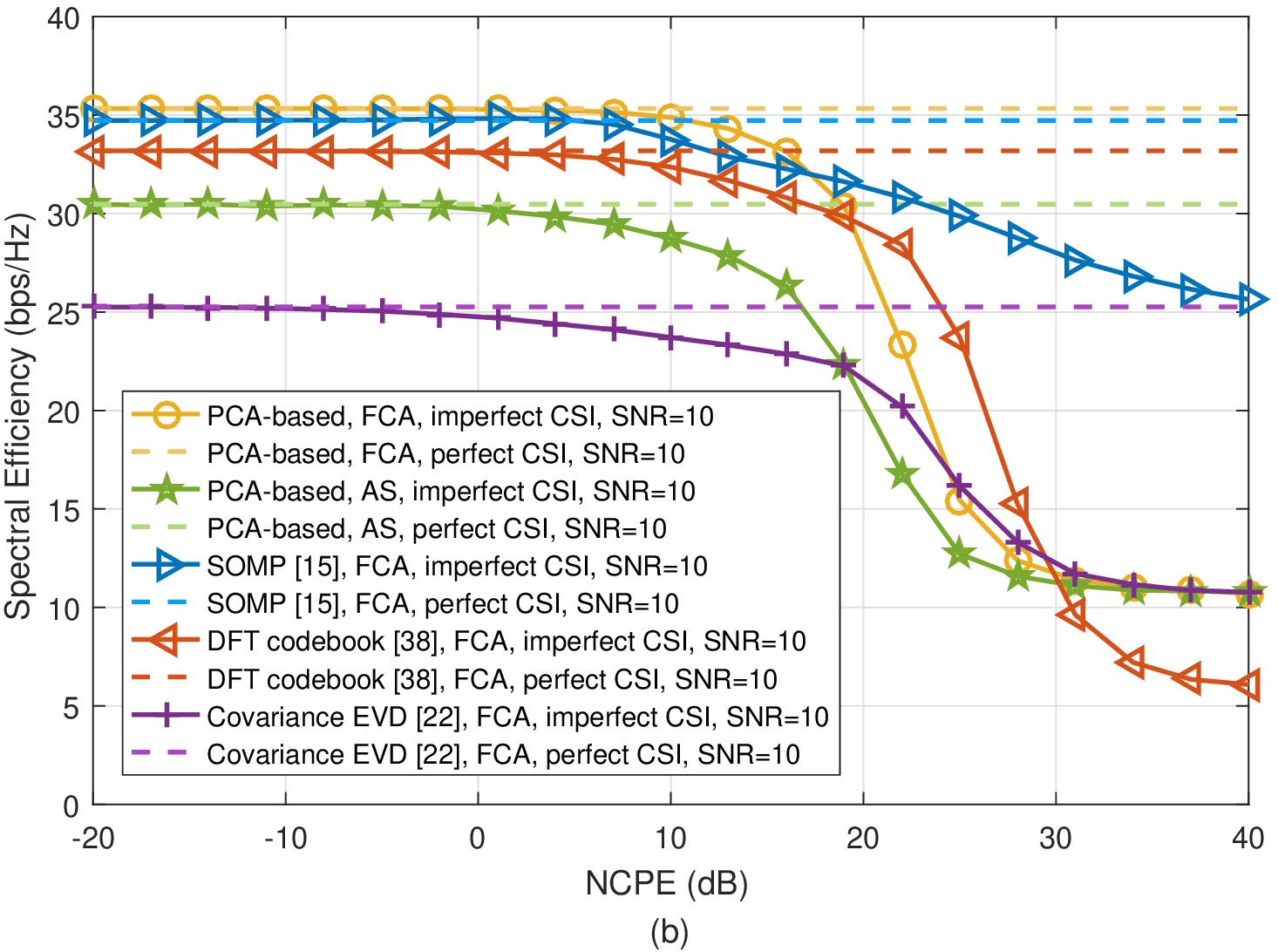}}
    	\caption{(a) SE of different hybrid precoder schemes under the hybrid transmit array and fully-digital receive array; (b) SE of different hybrid precoder/combiner under the hybrid transmit and hybrid receive arrays.}\label{fig_perturbation}
\end{figure}
Fig. \ref{fig_perturbation} (a) and (b) compare
 the robustness of different hybrid precoder/combiner schemes to imperfect CSI, which results from
  the channel perturbation including channel estimation error, the CSI quantization in channel feedback, and/or outdated CSI.
 We define the normalized channel perturbation error (NCPE) 
 as ${\rm NCPE}=\frac{\sum_{k=1}^{K}||\mathbf{H}[k]-\mathbf{H}_{\rm per}[k]||_F^2}{\sum_{k=1}^{K}||\mathbf{H}[k]||_F^2}$, where the imperfect CSI matrix is modeled as $\mathbf{H}_{\rm per}[k]=\mathbf{H}[k]+\mathbf{N}_{\rm per}[k]$, and we assume the entries of channel perturbation error $\mathbf{N}_{\rm per}[k]$ follows the independent and identically distributed complex Gaussian distribution $\mathcal{CN}(0,\sigma_{\rm per}^2)$.
 From Fig. \ref{fig_perturbation} (a), we can observe that both the proposed scheme and covariance EVD-based approach have the similar robustness to channel perturbation.
This is because both our proposed scheme and covariance EVD scheme can extract the correct frequency-flat component from the frequency-selective channels with perturbation.
 From Fig. \ref{fig_perturbation} (b), we can observe that both the proposed scheme and SOMP-based approach have the similar robustness to channel perturbation and outperform other schemes. Note that the SOMP-based approach requires the full knowledge of steering vectors, which can be impractical. When the full knowledge of steering vectors is inaccessible (SOMP-based scheme reduces to the DFT codebook scheme), the performance degrades drastically. By contrast, our proposed scheme does not require the knowledge of steering vectors and have the robustness performance to channel perturbation.


\section{Conclusions}
In this paper, we proposed a hybrid precoding scheme based on PCA for broadband mmWave massive MIMO
systems. We first designed a low-dimensional frequency-flat precoder/combiner from the optimal frequency-selective precoder/combiner based on PCA for fully-connected array. Moreover, we extended the proposed PCA-based hybrid precoder/combiner design to the partially-connected subarray given the antenna grouping pattern. 
For the adaptive subarray, we further proposed the shared-AHC algorithm inspired by cluster analysis in the field of machine learning to group the antennas for the further improved SE performance. Finally, the better SE, BER, and EE performance of the proposed hybrid precoder/combiner solution over state-of-the-art solutions was verified in simulations. 


%

\appendices
\section{Proof of Lemma 1}
To start with, we make the following system approximation.
\newtheorem{app}{\textbf{Approximation}}
\begin{app}
    We assume that the hybrid precoder $\mathbf{F}_{\rm RF}\mathbf{F}_{\rm BB}[k]$ can be sufficiently ``close" to the optimal fully-digital precoder ${\mathbf{F}_{\rm FD}^{\rm opt}}[k]=\mathbf{F}_{\rm RF}\mathbf{F}_{\rm BB}[k]$, $\forall k$, under the given system model and parameters (e.g., $N_t,N_r,N_t^{\rm RF},N_r^{\rm RF},N_s,N_{\rm cl},N_{\rm ray},...$). Define $\mathbf{\Sigma}[k]=\text{\rm blkdiag}(\mathbf{\Sigma}_1[k], \mathbf{\Sigma}_2[k])$, where $\mathbf{\Sigma}_1[k]=[\mathbf{\Sigma}[k]]_{1:N_s,1:N_s}$, and
      $\mathbf{V}[k]=\begin{bmatrix}
                   {\mathbf{F}_{\rm FD}^{\rm opt}}[k] & \mathbf{V}_2[k]
                 \end{bmatrix}$ in (\ref{SVD}), $\forall k$, this ``closeness" is defined based on the following two equivalent approximations:
    \begin{enumerate}[(1)]
      \item The eigenvalues of the matrix $\mathbf{I}_{N_s}-{\mathbf{F}_{\rm FD}^{\rm opt}}^H[k]\mathbf{F}_{\rm RF}\mathbf{F}_{\rm BB}[k]\mathbf{F}_{\rm BB}^H[k]\mathbf{F}_{\rm RF}^H{\mathbf{F}_{\rm FD}^{\rm opt}}[k]$ are small. In this case, it can be equivalently stated as ${\mathbf{F}_{\rm FD}^{\rm opt}}^H[k]\mathbf{F}_{\rm RF}\mathbf{F}_{\rm BB}[k]\approx\mathbf{I}_{N_s}$.
      \item The singular values of the matrix $\mathbf{V}_2^H[k]\mathbf{F}_{\rm RF}\mathbf{F}_{\rm BB}[k]$ are small, i.e. $\mathbf{V}_2^H[k]\mathbf{F}_{\rm RF}\mathbf{F}_{\rm BB}[k]\approx\mathbf{0}$.
    \end{enumerate}
\end{app}
According to (\ref{SVD}), the objective function of problem (\ref{opt}) can be written as
\begin{equation}\label{opt_1}
  \sum_{k=1}^{K}\log_2(\det(\mathbf{I}_{N_r}+\tfrac{1}{\sigma_n^2}\mathbf{\Sigma}^2[k]\mathbf{V}^H[k]\mathbf{F}_{\rm RF}\mathbf{F}_{\rm BB}[k]\mathbf{F}_{\rm BB}^H[k]\mathbf{F}_{\rm RF}^H\mathbf{V}[k])).
\end{equation}
Following the similar derivation for formula (12) in \cite{OMP}, we define (*).
Based on the above definition, (\ref{opt_1}) can be approximated as
\begin{equation}\label{opt_2}
  \sum_{k=1}^{K}\log_2(\det(\mathbf{I}_{N_r}+\frac{1}{\sigma_n^2}
\begin{bmatrix}
  \mathbf{\Sigma}_1^2[k] & \mathbf{0} \\
  \mathbf{0} & \mathbf{\Sigma}_2^2[k]
\end{bmatrix}\begin{bmatrix}
          \mathbf{M}_{11}[k] & \mathbf{M}_{12}[k] \\
          \mathbf{M}_{21}[k] & \mathbf{M}_{22}[k]
        \end{bmatrix})).
\end{equation}
\begin{figure*}[tb]
\begin{equation*}
\begin{aligned}
    &\mathbf{V}^H[k]\mathbf{F}_{\rm RF}\mathbf{F}_{\rm BB}[k]\mathbf{F}_{\rm BB}^H[k]\mathbf{F}_{\rm RF}^H\mathbf{V}[k]
    \\=&\!\!\begin{bmatrix}
          {\mathbf{F}_{\rm FD}^{\rm opt}}^H\![k]\mathbf{F}_{\rm RF}\mathbf{F}_{\rm BB}\![k]
          \mathbf{F}_{\rm BB}^H\![k]\mathbf{F}_{\rm RF}^H{\mathbf{F}_{\rm FD}^{\rm opt}}\![k] & {\mathbf{F}_{\rm FD}^{\rm opt}}^H\![k]\mathbf{F}_{\rm RF}\mathbf{F}_{\rm BB}\![k]
          \mathbf{F}_{\rm BB}^H\![k]\mathbf{F}_{\rm RF}^H\!\mathbf{V}_2\![k] \\
          \mathbf{V}_2^H\![k]\mathbf{F}_{\rm RF}\mathbf{F}_{\rm BB}\![k]
          \mathbf{F}_{\rm BB}^H\![k]\mathbf{F}_{\rm RF}^H{\mathbf{F}_{\rm FD}^{\rm opt}}\![k] & \mathbf{V}_2^H\![k]\mathbf{F}_{\rm RF}\mathbf{F}_{\rm BB}\![k]
          \mathbf{F}_{\rm BB}^H[k]\mathbf{F}_{\rm RF}^H\mathbf{V}_2\![k]
        \end{bmatrix}
    \!\!=\!\!\begin{bmatrix}
          \mathbf{M}_{11}\![k] & \mathbf{M}_{12}\![k] \\
          \mathbf{M}_{21}\![k] & \mathbf{M}_{22}\![k]
        \end{bmatrix}
\end{aligned}\tag{*}
\end{equation*}
\end{figure*}
According to Schur complement identity for matrix determinants, (\ref{opt_2}) is equivalent to
\begin{equation}\label{opt_3}
  \begin{aligned}
  &\sum_{k=1}^{K}(\log_2(\det(\mathbf{I}_{N_s}
  \\&+\tfrac{1}{\sigma_n^2}
\mathbf{\Sigma}_1^2[k]\mathbf{M}_{11}[k]))+\log_2(\det(\mathbf{I}_{N_r-N_s}+\tfrac{1}{\sigma_n^2}
\mathbf{\Sigma}_2^2[k]\mathbf{M}_{22}[k]
\\&-\tfrac{1}{\sigma_n^2}
\mathbf{\Sigma}_2^2[k]\mathbf{M}_{21}[k](\mathbf{I}_{N_s}+\tfrac{1}{\sigma_n^2}
\mathbf{\Sigma}_1^2[k]\mathbf{M}_{11}[k])^{-1}\mathbf{\Sigma}_1^2[k]\mathbf{M}_{12}[k]))).
  \end{aligned}
\end{equation}
According to {\it Approximation 1 (2)}, $\mathbf{M}_{21}[k]$, $\mathbf{M}_{12}[k]$, and $\mathbf{M}_{22}[k]$ are approximately 0, so (\ref{opt_3}) can be approximated as
\begin{equation}\label{opt_drv}
\begin{aligned}
    \sum_{k=1}^{K}&\log_2(\det(\mathbf{I}_{N_s}\!\!+\!\!\tfrac{1}{\sigma_n^2}
    \mathbf{\Sigma}_1^2[k]{\mathbf{F}_{\rm FD}^{\rm opt}}^H[k]\mathbf{F}_{\rm RF}\mathbf{F}_{\rm BB}[k]
    \mathbf{F}_{\rm BB}^H[k]
    \\&\times\mathbf{F}_{\rm RF}^H{\mathbf{F}_{\rm FD}^{\rm opt}}[k])),
\end{aligned}
\end{equation}
where the equation holds when Approximation 1 holds.
Based on $\mathbf{I}+\mathbf{B}\mathbf{A}=(\mathbf{I}+\mathbf{B})
      (\mathbf{I}-(\mathbf{I}+\mathbf{B})^{-1}\mathbf{B}(\mathbf{I}-\mathbf{A}))$ with the definition $\mathbf{B}=\frac{\mathbf{\Sigma}_1^2}{\sigma_n^2}$ and $\mathbf{A}={\mathbf{F}_{\rm FD}^{\rm opt}}^H[k]\mathbf{F}_{\rm RF}\mathbf{F}_{\rm BB}[k]\mathbf{F}_{\rm BB}^H[k]\mathbf{F}_{\rm RF}^H{\mathbf{F}_{\rm FD}^{\rm opt}}[k]$, (\ref{opt_drv}) is equivalent to
\begin{equation}\label{opt_d1}
  \begin{aligned}
    &\sum_{k=1}^{K}(\log_2(\det(\mathbf{I}_{N_s}+\tfrac{1}{\sigma_n^2}
\mathbf{\Sigma}_1^2[k]))
\\&+\log_2(\det(\mathbf{I}_{N_s}-(\mathbf{I}_{N_s}+\tfrac{1}{\sigma_n^2}
\mathbf{\Sigma}_1^2[k])^{-1}\tfrac{1}{\sigma_n^2}\mathbf{\Sigma}_1^2[k]
(\mathbf{I}_{N_s}
\\&-{\mathbf{F}_{\rm FD}^{\rm opt}}^H[k]\mathbf{F}_{\rm RF}\mathbf{F}_{\rm BB}[k]
\mathbf{F}_{\rm BB}^H[k]\mathbf{F}_{\rm RF}^H{\mathbf{F}_{\rm FD}^{\rm opt}}[k])))).
  \end{aligned}
\end{equation}
{\it Approximation 1 (1)} implies the eigenvalues of matrix $(\mathbf{I}_{N_s}+\mathbf{\Sigma}_1^2[k]/\sigma_n^2)^{-1}\mathbf{\Sigma}_1^2[k]/\sigma_n^2(\mathbf{I}_{N_s}-{\mathbf{F}_{\rm FD}^{\rm opt}}^H[k]\mathbf{F}_{\rm RF}\mathbf{F}_{\rm BB}[k]
\mathbf{F}_{\rm BB}^H[k]\mathbf{F}_{\rm RF}^H{\mathbf{F}_{\rm FD}^{\rm opt}}[k])$ are small.
So $\log_2(\det(\mathbf{I}_{N_s}-\mathbf{X}))\approx \log_2(1-\text{Tr}(\mathbf{X}))\approx -\text{Tr}(\mathbf{X})$. Thus (\ref{opt_d1}) can be approximated as
\begin{equation}\label{opt_d2}
  \begin{aligned}
     \sum_{k=1}^{K}&(\log_2(\det(\mathbf{I}_{N_s}\!\!+\!\!\tfrac{1}{\sigma_n^2}
\mathbf{\Sigma}_1^2[k]))
\\&-\text{Tr}((\mathbf{I}_{N_s}\!\!+\!\!\tfrac{1}{\sigma_n^2}
\mathbf{\Sigma}_1^2[k])^{-1}\tfrac{1}{\sigma_n^2}\mathbf{\Sigma}_1^2[k]
(\mathbf{I}_{N_s}
\\&\!\!-\!\!{\mathbf{F}_{\rm FD}^{\rm opt}}^H[k]\mathbf{F}_{\rm RF}\mathbf{F}_{\rm BB}[k]
\mathbf{F}_{\rm BB}^H[k]\mathbf{F}_{\rm RF}^H{\mathbf{F}_{\rm FD}^{\rm opt}}[k]))),
  \end{aligned}
\end{equation}
where the equation holds when Approximation 1 holds. Based on the high SNR approximation $(\mathbf{I}_{N_s}+\frac{\mathbf{\Sigma}_1^2}{\sigma_n^2})^{-1}\frac{\mathbf{\Sigma}_1^2}{\sigma_n^2}\approx \mathbf{I}_{N_s}$, (\ref{opt_d2}) can be further approximated as
\begin{equation}\label{opt_d3}
\begin{aligned}
  \sum_{k=1}^{K}&(\log_2(\det(\mathbf{I}_{N_s}+\tfrac{1}{\sigma_n^2}
  \mathbf{\Sigma}_1^2[k]))
  -(N_s
  \\&-||{\mathbf{F}_{\rm FD}^{\rm opt}}^H[k]\mathbf{F}_{\rm RF}\mathbf{F}_{\rm BB}[k]||_F^2)),
\end{aligned}
\end{equation}
where this equation holds at high SNR conditions.
Therefore, the optimization problem (\ref{opt}) can be approximated as (\ref{opt_final}).

\section{Proof of Proposition 1}
The SVD of $\mathbf{F}_{\rm RF}$ can be written as $\mathbf{U}_{\mathbf{F}_{\rm RF}}[\mathbf{\Sigma}_{\mathbf{F}_{\rm RF}} \ \mathbf{0}_{(N_t-N_t^{\rm RF})\times N_t^{\rm RF}}]^T\mathbf{V}_{\mathbf{F}_{\rm RF}}^H=\check{\mathbf{U}}_{\mathbf{F}_{\rm RF}}\mathbf{\Sigma}_{\mathbf{F}_{\rm RF}}\mathbf{V}_{\mathbf{F}_{\rm RF}}^H$, where $\mathbf{U}_{\mathbf{F}_{\rm RF}}\in\mathbb{C}^{N_t\times N_t}$, $\mathbf{\Sigma}_{\mathbf{F}_{\rm RF}}\in\mathbb{C}^{N_t^{\rm RF}
\times N_t^{\rm RF}}$, $\mathbf{V}_{\mathbf{F}_{\rm RF}}\in\mathbb{C}^{N_t^{\rm RF}
\times N_t^{\rm RF}}$, and $\check{\mathbf{U}}_{\mathbf{F}_{\rm RF}}=[\mathbf{U}_{\mathbf{F}_{\rm RF}}]_{:,1:N_t^{\rm RF}}$. The formula (\ref{F_BB}) can be further expressed as
$\mathbf{F}_{\rm BB}[k]
\!=\!\mathbf{V}_{\mathbf{F}_{\rm RF}}\mathbf{\Sigma}_{\mathbf{F}_{\rm RF}}^{-1}\mathbf{V}_{\mathbf{F}_{\rm RF}}^H\mathbf{\widetilde{F}}_{\rm BB}[k]$.
Hence the objective function in (\ref{opt_final}) can be written as
\begin{equation}\label{bef_item}
\sum\nolimits_{k=1}^{K}
||{\mathbf{F}_{\rm FD}^{\rm opt}}^H[k]\check{\mathbf{U}}_{\mathbf{F}_{\rm RF}}\mathbf{V}_{\mathbf{F}_{\rm RF}}^H\mathbf{\widetilde{F}}_{\rm BB}[k]||_F^2.
\end{equation}
According to previous work \cite{uni_cons}, unitary constraints offer a close performance to the total power constraint and provide a relatively simple form of solution. To simplify the problem, we consider the condition under unitary power constraints instead. Therefore, the water-filling power allocation coefficients can be ignored. Specifically, the equivalent baseband precoder is $\mathbf{\widetilde{F}}_{\rm BB}[k]=[\mathbf{\widetilde{V}}[k]]_{:,1:N_s}$. Hence, $\mathbf{\widetilde{F}}_{\rm BB}[k]$ is a unitary or semi-unitary matrix depending on $N_s=N_t^{\rm RF}$ or $N_s<N_t^{\rm RF}$. Therefore, in the following part, we discuss the two conditions separately.

When $N_s=N_t^{\rm RF}$, $\mathbf{\widetilde{F}}_{\rm BB}[k]$ is a unitary matrix. Therefore, (\ref{bef_item}) can be simplified as
\begin{equation}\label{s_eq_RF}
\begin{aligned}
&\sum\nolimits_{k=1}^{K}
||{\mathbf{F}_{\rm FD}^{\rm opt}}^H[k]\check{\mathbf{U}}_{\mathbf{F}_{\rm RF}}\mathbf{V}_{\mathbf{F}_{\rm RF}}^H\mathbf{\widetilde{F}}_{\rm BB}[k]||_F^2
\\&=\text{Tr}(\sum\nolimits_{k=1}^{K}\check{\mathbf{U}}_{\mathbf{F}_{\rm RF}}^H{\mathbf{F}_{\rm FD}^{\rm opt}}[k]{\mathbf{F}_{\rm FD}^{\rm opt}}^H[k]\check{\mathbf{U}}_{\mathbf{F}_{\rm RF}})
\\=&\!\text{Tr}(\!\begin{bmatrix}
\check{\mathbf{U}}_{\mathbf{F}_{\rm RF}}^H{\mathbf{F}_{\rm FD}^{\rm opt}}\![1] & \cdots & \check{\mathbf{U}}_{\mathbf{F}_{\rm RF}}^H{\mathbf{F}_{\rm FD}^{\rm opt}}\![K]
\end{bmatrix}\!\!\!\!
\begin{bmatrix}
{\mathbf{F}_{\rm FD}^{\rm opt}}^H\![1]\check{\mathbf{U}}_{\mathbf{F}_{\rm RF}} \\
\vdots \\
{\mathbf{F}_{\rm FD}^{\rm opt}}^H\![K]\check{\mathbf{U}}_{\mathbf{F}_{\rm RF}}
\end{bmatrix}\!\!)
\\=&\text{Tr}(\check{\mathbf{U}}_{\mathbf{F}_{\rm RF}}^H
{\widetilde{\mathbf{F}}_{\rm FD}^{\rm opt}}
{\widetilde{\mathbf{F}}_{\rm FD}^{\rm opt}}{}^H
\check{\mathbf{U}}_{\mathbf{F}_{\rm RF}})
\\=&\text{Tr}(\check{\mathbf{U}}_{\mathbf{F}_{\rm RF}}^H\mathbf{U}_{\widetilde{\mathbf{F}}_{\rm FD}^{\rm opt}}\mathbf{\Sigma}_{\widetilde{\mathbf{F}}_{\rm FD}^{\rm opt}}^2\mathbf{U}_{\widetilde{\mathbf{F}}_{\rm FD}^{\rm opt}}^H\check{\mathbf{U}}_{\mathbf{F}_{\rm RF}}).
\end{aligned}
\end{equation}
Since $\mathbf{U}_{\mathbf{F}_{\rm RF}}$ and $\mathbf{U}_F$ are semi-unitary and unitary matrix, (\ref{s_eq_RF}) reaches the maximum only when $\check{\mathbf{U}}_{\rm RF}=[\mathbf{U}_F]_{:,1:N_t^{\rm RF}}$. Moreover, the rank of $\mathbf{F}_{\rm RF}$ is $N_t^{\rm RF}$.
Hence, {the sub-optimal RF precoder can be expressed as
$\mathbf{F}_{\rm RF}=\frac{1}{{\sqrt {{N_t}} }}{\rm exp}(j\angle([\mathbf{U}_{\widetilde{\mathbf{F}}_{\rm FD}^{\rm opt}}]_{:,1:N_t^{\rm RF}}))$\footnote{This ``sub-optimal'' is due to the approximation by considering the CMC.}.}

When $N_s<N_t^{\rm RF}$, $\mathbf{\widetilde{F}}_{\rm BB}[k]$ is a semi-unitary matrix. Given the SVD $\mathbf{\widetilde{F}}_{\rm BB}[k]=\mathbf{U}_{\rm BB}[k]\left[\mathbf{I}_{N_s} \ \mathbf{0}\right]^T\\\times\mathbf{V}_{\rm BB}^H[k]$,
the objective function (\ref{bef_item}) can be simplified as
\begin{equation}\label{s_l_RF}
\begin{aligned}
\sum_{k=1}^{K}
&\text{Tr}({\mathbf{F}_{\rm FD}^{\rm opt}}^H[k]\check{\mathbf{U}}_{\mathbf{F}_{\rm RF}}\!\mathbf{V}_{\mathbf{F}_{\rm RF}}\!\mathbf{U}_{\rm BB}[k]\text{blkdiag}(\mathbf{I}_{N_s},\mathbf{0})\mathbf{U}_{\rm BB}^H[k]
\\&\times\mathbf{V}_{\mathbf{F}_{\rm RF}}^H\!
\check{\mathbf{U}}_{\mathbf{F}_{\rm RF}}^H\!{\mathbf{F}_{\rm FD}^{\rm opt}}[k]).
\end{aligned}
\end{equation}
It is obvious that the solution {maximizing} (\ref{s_eq_RF}) also {maximizing} (\ref{s_l_RF}). Therefore, following the similar derivation of (\ref{s_eq_RF}), the conclusion of $\mathbf{F}_{\rm RF}=\frac{1}{{\sqrt {{N_t}} }}{\rm exp}(j\angle([\mathbf{U}_{\widetilde{\mathbf{F}}_{\rm FD}^{\rm opt}}]_{:,1:N_t^{\rm RF}}))$ is easy to be reached.

\section{Proof of Proposition 2}
Substituting $\mathbf{W}^{\rm WLS}_{\rm BB}[k]$ (\ref{W_BB_LS}) into the objective function of (\ref{problem_W}), we obtain
\begin{equation}
\begin{aligned}
&\sum_{k=1}^{K}||\mathbb{E}[\mathbf{y}[k]\mathbf{y}^H[k]]^{\frac{1}{2}}({\mathbf{W}_{\rm FD}^{\rm opt}}[k]-\mathbf{W}_{\rm RF}\mathbf{W}_{\rm BB}[k])||_F^2
\\=&\!\sum_{k=1}^{K}\!||\mathbb{E}[\mathbf{y}[k]\mathbf{y}^H[k]]^{\frac{1}{2}}{\mathbf{W}_{\rm FD}^{\rm opt}}[k]-\mathbb{E}[\mathbf{y}[k]\mathbf{y}^H[k]]^{\frac{1}{2}}\mathbf{W}_{\rm RF}
\\&\times(\mathbf{W}_{\rm RF}^H\mathbb{E}[\mathbf{y}[k]\mathbf{y}^H[k]]\mathbf{W}_{\rm RF})^{-1}\!\mathbf{W}_{\rm RF}^H\mathbb{E}[\mathbf{y}[k]\mathbf{y}^H[k]]{\mathbf{W}_{\rm FD}^{\rm opt}}[k])||_F^2.
\end{aligned}
\end{equation}
Defining $\mathbf{A}[k]=\mathbb{E}[\mathbf{y}[k]\mathbf{y}^H[k]]^{\frac{1}{2}}{\mathbf{W}_{\rm FD}^{\rm opt}}[k]$ and $\mathbf{B}[k]=\mathbb{E}[\mathbf{y}[k]\mathbf{y}^H[k]]^{\frac{1}{2}}\mathbf{W}_{\rm RF}$, we further obtain
\begin{equation}
\begin{aligned}
&\sum_{k=1}^{K}||\mathbf{A}[k]-\mathbf{B}[k](\mathbf{B}^H[k]
\mathbf{B}[k])^{-1}\mathbf{B}^H[k]\mathbf{A}[k]||_F^2
\\=&\sum_{k=1}^{K}\text{Tr}(\mathbf{A}^H[k]\mathbf{A}[k])
\\&-\sum_{k=1}^{K}\text{Tr}(\mathbf{A}^H[k]\mathbf{B}[k](\mathbf{B}^H[k]
\mathbf{B}[k])^{-1}\mathbf{B}^H[k]\mathbf{A}[k]).
\end{aligned}
\end{equation}
Hence, the minimization problem can be formulated as the following maximization problem
\begin{equation}\label{max_W}
\setlength{\abovedisplayskip}{4pt}
\setlength{\belowdisplayskip}{4pt}
\begin{aligned}
\max_{\mathbf{W}_{\rm RF},\mathbf{W}_{\rm BB}[k]}&\sum\nolimits_{k=1}^{K}\text{Tr}(\mathbf{A}^H[k]\mathbf{B}[k](\mathbf{B}^H[k]
\mathbf{B}[k])^{-1}\mathbf{B}^H[k]\mathbf{A}[k])
\\&\text{s.t. }\mathbf{B}[k]=\mathbb{E}[\mathbf{y}[k]\mathbf{y}^H[k]]^{\frac{1}{2}}\mathbf{W}_{\rm RF},\mathbf{W}_{\rm RF}\in\mathcal{W}_{\rm RF}.
\end{aligned}
\end{equation}
{The SVD of $\mathbf{W}_{\rm RF}$ can be written as $\mathbf{U}_{\mathbf{W}_{\rm RF}}[\mathbf{\Sigma}_{\mathbf{W}_{\rm RF}} \ \mathbf{0}_{(N_r-N_r^{\rm RF})\times N_r^{\rm RF}}]^T\mathbf{V}_{\mathbf{W}_{\rm RF}}^H=\check{\mathbf{U}}_{\mathbf{W}_{\rm RF}}\mathbf{\Sigma}_{\mathbf{W}_{\rm RF}}\mathbf{V}_{\mathbf{W}_{\rm RF}}^H$, where $\mathbf{U}_{\mathbf{W}_{\rm RF}}\in\mathbb{C}^{N_r\times N_r}$, $\mathbf{\Sigma}_{\mathbf{W}_{\rm RF}}\in\mathbb{C}^{N_r^{\rm RF}
\times N_r^{\rm RF}}$, $\mathbf{V}_{\mathbf{W}_{\rm RF}}\in\mathbb{C}^{N_r^{\rm RF}
\times N_r^{\rm RF}}$, and $\check{\mathbf{U}}_{\mathbf{W}_{\rm RF}}=[\mathbf{U}_{\mathbf{W}_{\rm RF}}]_{:,1:N_r^{\rm RF}}$. Similarly, the SVD of $\mathbf{B}[k]$ can be written as $\mathbf{U}_B[k][\mathbf{\Sigma}_B[k] \ \mathbf{0}_{(N_r\!-\!N_r^{\rm RF})\!\times\! N_r^{\rm RF}}]^T\mathbf{V}_B[k]^H\!\!=\!\!\check{\mathbf{U}}_B[k]
\mathbf{\Sigma}_B[k]\\\times\mathbf{V}_B[k]^H$, where $\mathbf{U}_B[k]\in\mathbb{C}^{N_r\times N_r}$, $\mathbf{\Sigma}_B[k]\in\mathbb{C}^{N_r^{\rm RF}
\times N_r^{\rm RF}}$, $\mathbf{V}_B[k]\in\mathbb{C}^{N_r^{\rm RF}
\times N_r^{\rm RF}}$, and $\check{\mathbf{U}}_B[k]=[\mathbf{U}_B[k]]_{:,1:N_r^{\rm RF}}$. Substituting $\mathbf{B}[k]$ with its SVD and comparing with (\ref{s_eq_RF}), the objective function of the problem (\ref{max_W}) can be further simplified as
\begin{equation}\label{max_W_2}
\begin{aligned}
  &\sum_{k=1}^{K}\text{Tr}(\mathbf{A}^H[k]\mathbf{B}[k](\mathbf{B}^H[k]
\mathbf{B}[k])^{-1}\mathbf{B}^H[k]\mathbf{A}\![k])
\\\!\!=&\!\!\sum_{k=1}^{K}\!||\hat{\mathbf{U}}_B^H\![k]\mathbf{A}\![k]||_F^2
\!=\!\text{Tr}(\check{\mathbf{U}}_B^H\!
\mathbf{U}_W\!\mathbf{\Sigma}_W^2\!
\mathbf{U}_W^H\check{\mathbf{U}}_B\!),
\end{aligned}
\end{equation}
where $\mathbf{U}_W\mathbf{\Sigma}_W\mathbf{V}_W^H$ is the SVD of $\mathbf{W}$ and $\mathbf{W}=\begin{bmatrix}
    \mathbf{A}[1] \ \cdots \ \mathbf{A}[K]
	\end{bmatrix}$. Thus the optimal $\check{\mathbf{U}}_B=[\mathbf{U}_W]_{:,1:N_r^{\rm RF}}$.
We can find a unitary matrix $\mathbf{U}_R[k]\in\mathbb{C}^{N_r
\times N_r}$ satisfying $\mathbf{U}_B[k]=\mathbf{U}_R[k]\mathbf{U}_{\mathbf{W}_{\rm RF}}$, thus $\hat{\mathbf{U}}_{\mathbf{W}_{\rm RF}}=\mathbf{U}_R^H[k]\hat{\mathbf{U}}_B[k]=\mathbf{U}_R^H[k][\mathbf{U}_W]_{:,1:N_r^{\rm RF}}$. Hence, a sub-optimal solution to problem (\ref{max_W}) is $\mathbf{W}_{\rm RF}=\frac{1}{{\sqrt {{N_r}} }}{\rm exp}(j\angle([\mathbf{U}_W]_{:,1:N_r^{\rm RF}}))$.}
\vspace{-3mm}
\section{Proof of (\ref{w_ass})}
By substituting (\ref{SVD}) and (\ref{F_BB}) into (\ref{H_eff}), we obtain
    $\mathbf{H}_{\rm eff}\![k]\!=\!\mathbf{H}\![k]\mathbf{F}_{\rm RF}\mathbf{F}_{\rm BB}\![k]
    \!=\!\mathbf{U}\![k]\mathbf{\widetilde{U}}\![k]\mathbf{\widetilde{\Sigma}}\![k]\mathbf{\widetilde{V}}^H\![k]
    \\\times[\mathbf{\widetilde{V}}\![k]]_{:,1:N_s}\!\mathbf{\Lambda}\![k]$.
Defining $[\mathbf{\widetilde{V}}[k]]_{:,1:N_s}=\mathbf{\widetilde{V}}_{N_s}[k]$ and $\mathbf{\widetilde{V}}[k]=\begin{bmatrix}
\mathbf{\widetilde{V}}_{N_s}[k] & \mathbf{\widetilde{V}}_0[k]
\end{bmatrix}$, $\mathbf{H}_{\rm eff}[k]$ is
\begin{equation}\label{H_eff_AC}
\begin{aligned}
    \mathbf{H}_{\rm eff}[k]&=\mathbf{U}[k]\mathbf{\widetilde{U}}[k]\mathbf{\widetilde{\Sigma}}[k]\!\begin{bmatrix}
                                                                                   \mathbf{\widetilde{V}}_{N_s}^H[k] \\
                                                                                   \mathbf{\widetilde{V}}_{0}^H[k]
                                                                                 \end{bmatrix}\!\mathbf{\widetilde{V}}_{N_s}[k]\mathbf{\Lambda}[k]
    \\&=[\mathbf{U}[k]\mathbf{\widetilde{U}}[k]\mathbf{\widetilde{\Sigma}}[k]]_{:,1:N_s}\mathbf{\Lambda}[k].
\end{aligned}
\end{equation}
Combine (\ref{wf_1}) with (\ref{H_eff_AC}), it arrives
\begin{equation}
\begin{aligned}
    \mathbf{H}_{\rm eff}[k]&=[\mathbf{U}[k]\mathbf{\widetilde{U}}[k]\mathbf{\widetilde{\Sigma}}[k]]_{:,1:N_s}\mathbf{\Lambda}[k]
   \\&=[\mathbf{U}[k]\mathbf{\widetilde{U}}[k]]_{:,1:N_s}(\mu[\mathbf{\widetilde{\Sigma}}\![k]]_{1:N_s,1:N_s}-N_s\mathbf{I}_{N_s}).
\end{aligned}
\end{equation}
In large antenna, $\mathbf{\widetilde{U}}[k]=\mathbf{I}$ and $\mathbf{\widetilde{\Sigma}}[k]=\mathbf{\Sigma}[k]$. The effective channel can be further written as
\begin{equation}
\begin{aligned}
    \mathbf{H}_{\rm eff}[k]&=[\mathbf{U}[k]]_{:,1:N_s}(\mu[\mathbf{\Sigma}[k]]_{1:N_s,1:N_s}-N_s\mathbf{I}_{N_s})
    \\&=[\mathbf{A}_r[k]]_{:,1:N_s}(\mu[|\mathbf{P}[k]|]_{1:N_s,1:N_s}-N_s\mathbf{I}_{N_s}).
\end{aligned}
\end{equation}
Considering (\ref{blk_y}), the effective channel matrix $\mathbf{H}_{{\rm eff},\mathcal{T}_r}[k]=[\mathbf{A}_r]_{\mathcal{T}_r,1:N_s}(\mu[|\mathbf{P}[k]|]_{1:N_s,1:N_s}-N_s\mathbf{I}_{N_s})$.
Therefore, the covariance matrix of $\mathbf{y}_{\mathcal{T}_r}[k]$ can be expressed as
\begin{equation}
\begin{aligned}
    \mathbb{E}[\mathbf{y}_{\mathcal{T}_r}\![k]\mathbf{y}_{\mathcal{T}_r}^H\![k]]=&
    \mathbf{H}_{{\rm eff},\mathcal{T}_r}\![k]\mathbf{H}_{{\rm eff},\mathcal{T}_r}^H\![k]+\sigma_n^2\mathbf{I}_{N_r^{\rm sub}}
    \\=&
    [\mathbf{A}_r]_{\mathcal{T}_r,1:N_s}\!(\mu[|\mathbf{P}\![k]|]_{1:N_s,1:N_s}\!-\!N_s\mathbf{I}_{N_s}\!)^2
    \\&\times\!([\mathbf{A}_r]_{\mathcal{T}_r,1:N_s}\!)^H\!+\sigma_n^2\mathbf{I}_{N_r^{\rm sub}}.
\end{aligned}
\end{equation}
Since $|\mathbf{P}[k]|$ remains unchanged for different $k$, $\mathbb{E}[\mathbf{y}_{\mathcal{T}_r}[k]\mathbf{y}_{\mathcal{T}_r}^H[k]]$ is irrelevant with $k$ for large $N_t$.
Therefore, (\ref{w_ass}) is valid in the regime of very large number of antennas, and this approximation error can usually be negligible for mmWave massive MIMO with large number of antennas.


\ifCLASSOPTIONcaptionsoff
  \newpage
\fi

\end{document}